\documentclass[a4paper, 10pt, oneside, onecolumn,3p]{elsarticle}
\journal{Annals of Physics}
\usepackage{amsmath,color}
\usepackage{hyperref}
\usepackage{bm}
\usepackage{amssymb} 
\def\bra#1{\mathinner{\langle{#1}|}}
\def\ket#1{\mathinner{|{#1}\rangle}}
\def\coloneq{\mathrel{\mathop:}=}
\DeclareMathOperator{\tr}{Tr}
\DeclareMathOperator{\rank}{rnk}
\DeclareMathOperator{\sgn}{sgn}
\newcommand{\Id}{\mathbf{1}}
\newtheorem{theorem}{Theorem}
\newtheorem{lemma}{Lemma}
\newdefinition{remark}{Remark}
\newproof{proof}{Proof}
\newdefinition{definition}{Definition}

\biboptions{sort&compress}

\begin{document}

\title{Entanglement universality of two-qubit X-states}
\author[uq]{Paulo E. M. F. Mendon\c{c}a\corref{cor1}\fnref{fn1}}
\ead{pmendonca@gmail.com}
\author[jabuca]{Marcelo A. Marchiolli}
\ead{marcelo{\_}march@bol.com.br}
\author[ift]{Di\'ogenes Galetti}
\ead{galetti@ift.unesp.br}
\address[uq]{ARC Centre for Engineered Quantum Systems, School of Mathematics and Physics, The University of Queensland, St.
Lucia, Queensland 4072, Australia}
\address[jabuca]{Avenida General Os\'orio 414, centro, 14.870-100 Jaboticabal, SP, Brazil}
\address[ift]{Instituto de F\'isica Te\'orica, Universidade Estadual Paulista, Rua Dr. Bento Teobaldo Ferraz 271, Bloco II, Barra
Funda, 01140-070 S\~{a}o Paulo, SP, Brazil}
\cortext[cor1]{Corresponding author}
\fntext[fn1]{Permanent address: Academia da For\c{c}a A\'{e}rea, C.P. 970, 13.643-970 Pirassununga, SP, Brazil}
\date{\today}

\begin{abstract}
We demonstrate that for every two-qubit state there is a \emph{X-counterpart}, i.e., a corresponding two-qubit X-state of same spectrum and entanglement, as
measured by concurrence, negativity or relative entropy of entanglement. By parametrizing the set of two-qubit X-states and a
family of unitary transformations that preserve the sparse structure of a two-qubit X-state density matrix, we obtain the parametric form of a unitary transformation that
converts arbitrary two-qubit states into their X-counterparts. Moreover, we provide a semi-analytic prescription on how to set the
parameters of this unitary transformation in order to preserve concurrence or negativity. We also explicitly construct a set of
X-state density matrices, parametrized by their purity and concurrence, whose elements are in one-to-one correspondence with the points
of the concurrence versus purity (CP) diagram for generic two-qubit states.
\end{abstract}

\maketitle

\section{Introduction}

Despite our limited understanding of what entanglement is at the most fundamental level, many tasks that feature entanglement as a \emph{sine qua non} condition have been successfully performed thanks to our ever-growing ability to manipulate quantum systems comprised of interacting subsystems~\cite{91Ekert661,07Ursin481,92Bennett2881,08Barreiro282,93Bennett1895,12Ma269,97Shor1484,12Lopez773}. Ultimately, entanglement is an attribute of quantum states and, as such, practical applications will unavoidably rely upon one's ability to prepare certain density matrices. In practice, though, depending on the details of a particular implementation and on the types of noise that affect the relevant quantum system, some entangled states may turn out to be very hard to produce, thus limiting the entanglement available to practical applications. As a result, it is natural to ask: how much entanglement is left as we avoid certain density matrices?

In this paper this question is approached in the context of two-qubit states and with a clear specification as to which states are to be avoided. Surprisingly, we find that no entanglement (as quantified by three entanglement measures and with respect to a fixed level of mixedness) is lost as we avoid every two-qubit density matrix, but the sparse family that, in the computational basis, can display non-zero entries only along the main- and anti-diagonals; the so-called X-states~\cite{07Yu459}. For entanglement measures we consider concurrence, negativity and relative entropy of entanglement, in terms of which our main result acquires its more precise expression:  \emph{for every two-qubit state with a value of entanglement set by any of these measures, there is a corresponding X-state of same spectrum and same entanglement}.

Two-qubit X-states generalize many renowned families of entangled two-qubit states, for example, Bell states~\cite{00Nielsen}, Werner states~\cite{89Werner4277}, isotropic states~\cite{99Horodecki4206} and maximally entangled mixed states~\cite{00Ishizaka22310,01Verstraete12316,01Munro30302,03Wei22110}. They were first identified as a class of states of interest in the work of Yu and Eberly~\cite{07Yu459}, where some of their properties in connection with the phenomenon of sudden death of entanglement were investigated. Ever since, the interest in X-states exceeded its original motivation and has been manifested in many other contexts~\cite{05Retzker050504,09Rau412002,12Quesada1322,13Hedemann,13Costa}. Particularly relevant for this paper is the work of Hedemann~\cite{13Hedemann}, who provided compelling numerical evidence that the set of two-qubit X-states alone is sufficient to access every possible combination of concurrence and purity available to two-qubit states, and conjectured that any generic two-qubit state can be converted into a X-state via a unitary transformation that preserves concurrence. Besides proving Hedemann's conjecture, we demonstrate that it also holds true when entanglement is quantified with negativity or relative entropy of entanglement instead of concurrence.

Also closely related to our purposes is the work of Verstraete~\emph{et al.}~\cite{01Verstraete12316}, where it was shown that, for a fixed set of eigenvalues, the states of maximal concurrence, negativity or relative entropy of entanglement are the same X-states, thus establishing the top frontier of the relevant entanglement versus mixedness diagrams as comprised by X-states. Our main result extends theirs in implying that X-states not only border such diagrams, but can be put in a many-to-one correspondence with every internal point. 

From a pragmatic viewpoint, the interest in this universality property of X-states relies upon their inherent easiness of manipulation, both theoretical and experimental. Owing to the highly sparse form of X-state density matrices written in the computational basis (X-density matrices, for short), a great deal of symbolic computations is possible, even in the context of entanglement quantification where one is usually forced to resort to numerical approaches. The possibility of replacing generic two-qubit density matrices with X-density matrices is a promising route toward a deeper understanding of mixed-state entanglement. On the experimental side, two-qubit X-states can be produced and evolved, for example, with standard interactions arising in the context of nuclear magnetic resonance~\cite{09Rau412002,00Rau032301} and with variations of available technology for generating Werner states in optical and atomic implementations~\cite{02Zhang062315,04Barbieri177901,04Cinelli022321,04Peters133601,06Agarwal022315,13Jin2830}.

Throughout, aiming to take full advantage of the highly sparse form of two-qubit X-density matrices, we exploit the luxury of working in a constructive-analytic fashion. Largely, this is enabled by the introduction of a simple parametrization on the set of X-states, which leads to a geometric visualization of separable, entangled and rank-specific X-states in the relevant parameter space. Thanks to this, we are able to explicitly construct a set of two-qubit X-states that can be put in a one-to-one correspondence with the points of the CP-diagram for generic two-qubit states. Most importantly, we parametrize a unitary transformation that maps an arbitrary two-qubit state into a X-state of same entanglement (according to any one of the three considered measures), and show how to set the parameter values to achieve conservation of concurrence or negativity.

Our paper is structured as follows. In order to obtain the constructions that form the core of our work, in Sec.~\ref{sec:parametrize} we parametrize separable, entangled and rank-specific two-qubit X-density matrices. Our parametrizations are first put into use in Sec.~\ref{sec:minimalset}, where we explicitly construct a minimal set of X-states that exhausts the two-qubit CP-diagram. In Sec.~\ref{sec:universality} our main universality result is established by showing that every X-state can be disentangled with a unitary transformation that preserves the sparse structure of a two-qubit X-density matrix (Sec.~\ref{sec:disentanglement}) and that our selected entanglement measures vary continuously during the disentangling process (Sec.~\ref{sec:continuity}). We summarize our main results and discuss some possible avenues for future research in Sec.~\ref{sec:concludingremarks}.

\section{Parametrizing two-qubit X-states}\label{sec:parametrize}

Two-qubit X-states are quantum states of a four-dimensional Hilbert space that do not mix the subspaces $S_{1} = {\rm Span} (\ket{00}, \ket{11})$ and $S_{2} = {\rm Span} (\ket{01}, \ket{10})$. In the computational basis $\{\ket{00}, \ket{01}, \ket{10}, \ket{11}\}$, they assume the matrix form

\begin{equation}
\label{eq:Xstates_param}
\left[ \begin{array}{cccc}
\cos^{2}{\theta} & \cdot & \cdot & \sqrt{x}\, {\rm e}^{i \mu} \\
\cdot & \sin^{2}{\theta} \cos^{2}{\varphi} & \sqrt{y}\, {\rm e}^{i \nu} & \cdot \\
\cdot & \sqrt{y}\, {\rm e}^{-i \nu} & \sin^{2}{\theta} \sin^{2}{\varphi} \cos^{2}{\psi} & \cdot \\
\sqrt{x}\, {\rm e}^{-i \mu} & \cdot & \cdot & \sin^{2}{\theta} \sin^{2}{\varphi} \sin^{2}{\psi}
\end{array} \right]
\end{equation}
with $\theta,\varphi,\psi \in [0,\pi/2]$, $x,y \geq 0$ and  $\mu,\nu \in [0,2\pi]$. In order to highlight the resemblance of matrix~(\ref{eq:Xstates_param}) with the alphabet letter `X' (which justifies the nomenclature ``X-state''), we replace every vanishing entry of a matrix with a dot. Throughout, every density matrix of the form~(\ref{eq:Xstates_param}) is referred to as a \emph{X-density matrix}. More generally, every matrix possessing non-zero terms only along the main- and anti-diagonals is said to be of the \emph{X-form}.

That any X-density matrix has the form (\ref{eq:Xstates_param}) is a direct consequence of the fact that, apart from the
decoupling between $S_{1}$ and $S_{2}$, all the inbuilt constraints of (\ref{eq:Xstates_param}) are \emph{necessary} features 
of a density matrix: the parametrization along the main diagonal establishes only normalization and non-negativity of the 
diagonal entries, whereas the parametrization along the anti-diagonal establishes only Hermiticity. 

However, not every matrix of the form (\ref{eq:Xstates_param}) with $\theta,\varphi,\psi \in [0,\pi/2]$, $x,y \geq 0$ and  
$\mu,\nu \in [0,2\pi]$ is a density matrix. In what follows we show how to further constrain the ranges of $x$ and $y$ in order 
to make the set of matrices of the form (\ref{eq:Xstates_param}) with the corresponding parameter ranges to \emph{coincide} with 
the set of (i) X-density matrices, (ii) X-density matrices of a fixed rank and (iii) separable X-density matrices.

\subsection{Parametrizing two-qubit X-density matrices}\label{sec:param_xdensmat}

The set of two-qubit X-density matrices is equal to the subset of matrices of the form~(\ref{eq:Xstates_param}) with parameter
values that render it positive semidefinite. For that, we start by considering the characteristic equation for
(\ref{eq:Xstates_param}):
\begin{equation}
\label{eq:chareq}
\lambda^{4} - \mathfrak{a}_{1} \lambda^{3} + \mathfrak{a}_{2} \lambda^{2} - \mathfrak{a}_{3} \lambda + \mathfrak{a}_{4} = 0 \, ,
\end{equation}
where
\begin{equation}
\begin{array}{rclrcl}
\mathfrak{a}_{1} &=& 1 \, ,  & \qquad \mathfrak{a}_{3} &=& \mathcal{BH} + \mathcal{CG} - x \mathcal{B} - y \mathcal{C} \, , \\
\mathfrak{a}_{2} &=& \mathcal{BC} + \mathcal{G} + \mathcal{H} - x - y \, , & \mathfrak{a}_{4} &=& \mathcal{HG} - y \mathcal{H} -
x \mathcal{G} + xy \, .
\end{array}
\end{equation}
In the above, the calligraphic letters $\mathcal{B}$, $\mathcal{C}$, $\mathcal{G}$ and $\mathcal{H}$ are functions of the diagonal
parameters $\theta$, $\varphi$ and $\psi$. In fact, $\mathcal{C}$ and $\mathcal{B}$ ($\mathcal{H}$ and $\mathcal{G}$) give the 
sum (product) of the diagonal entries of the unnormalized density matrices of the `fictitious qubits' living in the subspaces
$S_{1}$ and $S_{2}$, respectively. Explicitly, $\mathcal{C} \coloneq 1 - \mathcal{B}$, 

\begin{equation}\label{eq:calBCGH}
\mathcal{B} \coloneq \sin^{2}{\theta} ( 1 - \sin^{2}{\varphi} \sin^{2}{\psi} ) \, ,\quad \mathcal{G} \coloneq \sin^{4}{\theta} \sin^{2}{\varphi} \cos^{2}{\varphi} \cos^{2}{\psi}\quad\mbox{and}\quad \mathcal{H} \coloneq \sin^{2}{\theta} \cos^{2}{\theta} \sin^{2}{\varphi} \sin^{2}{\psi} \,.
\end{equation}

Since the positive semidefiniteness of (\ref{eq:Xstates_param}) is equivalent to the set of inequalities 
$\{ \mathfrak{a}_{i} \geq 0 \}_{i=1,\dots,4}$~\cite{09Bernstein}, we are left with three nonvacuous inequalities
\begin{subequations}
\label{eq:physical}
\begin{align}
\label{eq:physical1}
\mathcal{BC} + (\mathcal{H}-x) + (\mathcal{G}-y) & \geq 0 \, , \\
\label{eq:physical2}
\mathcal{B}(\mathcal{H}-x) + \mathcal{C}(\mathcal{G}-y) & \geq 0 \, , \\
\label{eq:physical3}
(\mathcal{H}-x)(\mathcal{G}-y) & \geq 0 \, .
\end{align}
\end{subequations}
Now, due to the non-negativity of $\mathcal{B}$ and $\mathcal{C}$, it is clear that the inequalities above are simultaneously
satisfied if and only if
\begin{equation}
\label{eq:rangesxy}
x \in [0,\mathcal{H}] \quad \mbox{and} \quad y \in [0,\mathcal{G}] \, ,
\end{equation}
which summarize necessary and sufficient conditions for the positive semidefiniteness of the form (\ref{eq:Xstates_param}).
Therefore, the set of these matrices with $\theta, \varphi, \psi \in [0,\pi/2]$, $\mu, \nu \in [0,2\pi]$, $x \in [0,\mathcal{H}]$
and $y \in [0,\mathcal{G}]$ fully characterizes the set of two-qubit X-density matrices. 
As we shall see next, such a parametrization enables an appealing geometric visualization of two-qubit X-states and can be easily specialized to parametrize separable and fixed-rank two-qubit X-states.

\subsection{Parametrizing two-qubit X-density matrices of a fixed rank}

According to the Newton-Girard formulae~\cite{09Bernstein}, the coefficients $\mathfrak{a}_{i}$ of the characteristic equation (\ref{eq:chareq}) are the sum of all products of $i$ eigenvalues of matrix (\ref{eq:Xstates_param}). This observation
can be used to parametrize two-qubit X-density matrices with a fixed rank.

\begin{description}

\item[Rank-1:] Three zero eigenvalues impose $\mathfrak{a}_{2} = \mathfrak{a}_{3} = \mathfrak{a}_{4} = 0$ or, equivalently, the
saturation of Eqs. \eqref{eq:physical1}-\eqref{eq:physical3}. Clearly, this occurs if and only if
$x = \mathcal{H}$, $y = \mathcal{G}$ and $\mathcal{BC} = 0$, which can be recast as the logical disjunction:\footnote{The logical
equivalence between $( x = \mathcal{H}, y = \mathcal{G}, \mathcal{BC} = 0 )$ and (\ref{eq:rank1_params}) is established by the
easy-to-check implications: $\mathcal{B} = 0 \Rightarrow \mathcal{G} = 0$ and $\mathcal{C} = 0 \Rightarrow \mathcal{H} = 0$.}
\begin{equation}
\label{eq:rank1_params}
\left( x = \mathcal{H}, y = 0, \mathcal{B} = 0 \right) \quad \mbox{or} \quad \left( x = 0, y = \mathcal{G}, \mathcal{C} = 0
\right) \, .
\end{equation}

\item[Rank-2:] Two zero eigenvalues imply in $\mathfrak{a}_{2} > 0$ and $\mathfrak{a}_{3} = \mathfrak{a}_{4} = 0$, which means
that Eqs. \eqref{eq:physical2} and \eqref{eq:physical3} must be saturated, whereas \eqref{eq:physical1} must not. In this case,
some simple analysis shows that the following logical disjunction comprises all the possibilities:
\begin{equation}
\label{eq:rank2_params}
\left( x < \mathcal{H}, y=0, \mathcal{B} = 0 \right) \quad \mbox{or} \quad \left( x = 0, y < \mathcal{G}, \mathcal{C} = 0 \right)
\quad \mbox{or} \quad \left( x = \mathcal{H}, y = \mathcal{G}, \mathcal{BC} > 0 \right)\,. 
\end{equation}

\item[Rank-3:] The single zero eigenvalue imposes $\mathfrak{a}_{2} > 0$, $\mathfrak{a}_{3} > 0$ and $\mathfrak{a}_{4} = 0$, which
implies in the sole saturation of inequality \eqref{eq:physical3} or, equivalently,
\begin{equation}
\label{eq:rank3_params}
\left( x < \mathcal{H}, y = \mathcal{G}, \mathcal{B} > 0 \right) \quad \mbox{or} \quad \left( x = \mathcal{H}, y < \mathcal{G},
\mathcal{C} > 0 \right) \, .
\end{equation}

\item[Rank-4:] The absence of zero eigenvalues produces $\mathfrak{a}_{2} > 0$, $\mathfrak{a}_{3} > 0$ and $\mathfrak{a}_{4} > 0$,
which amounts to be the same as preventing saturation of inequalities \eqref{eq:physical1}-\eqref{eq:physical3}. This is equivalent to require
\begin{equation}
\label{eq:rank4_params}
\left( x < \mathcal{H}, y < \mathcal{G}, \mathcal{BC} > 0 \right) \, .
\end{equation}
\end{description}

It is thus clear that the set of two-qubit X-density matrices of a fixed rank is equivalent to the set of matrices (\ref{eq:Xstates_param}) with parameter values verifying the corresponding constraint specified above. Throughout, we shall refer to each
alternative of rank-specific parameter choice as a \emph{kind} of X-state. Accordingly, there are two kinds of rank-1 and rank-3
X-states, three kinds of rank-2 X-states and a single kind of rank-4 X-states.

\subsection{Parametrizing two-qubit separable X-density matrices}\label{sec:separable}

According to the PPT criterion \cite{96Peres1413,96Horodecki1}, the set of two-qubit X-density matrices is equal to the set of
matrices of the form (\ref{eq:Xstates_param}) with parameter values that render itself and its partial transpose positive
semidefinite. In Sec. \ref{sec:param_xdensmat}, we have seen how the positive semidefiniteness of (\ref{eq:Xstates_param})
constrains $x$ and $y$ [cf. Eq. (\ref{eq:rangesxy})]. In this section, we find analogous constraints for the positive
semidefiniteness of the partial transpose of (\ref{eq:Xstates_param}).

It suffices to consider the partial transpose over one of the two subsystems, which we choose to be the second. In that case, the
partial transpose operation over (\ref{eq:Xstates_param}) yields a matrix of the same form, but with $x$ and $y$ (and also $\mu$
and $\nu$) swapped over, which implies that the positive semidefiniteness of the partially transposed matrix is guaranteed by 
the constraints (\ref{eq:rangesxy}) with $x \in [0,\mathcal{G}]$ and $y \in [0,\mathcal{H}]$. Clearly, in order to have both
(\ref{eq:Xstates_param}) and its partial transpose positive semidefinite, $x$ and $y$ must be chosen according to
\begin{equation}
\label{eq:sep_params}
x \in [0,\min(\mathcal{G},\mathcal{H})] \quad \mbox{and} \quad y \in [0,\min(\mathcal{G},\mathcal{H})] \, .
\end{equation}
Therefore, the set of separable two-qubit X-density matrices is identical to the set of matrices of the form
(\ref{eq:Xstates_param}) with parameter values that verify (\ref{eq:sep_params}).

The results of this section are all summarized in Fig. \ref{fig:Xparameterspace}, which represent X-states with a fixed value of
$\mathcal{H} + \mathcal{G} + \mathcal{BC} = s$ in a $xy$ parameter space. Each plot corresponds to a different contribution of the
parameters $\mathcal{H}$, $\mathcal{G}$ and $\mathcal{BC}$ toward $s$, in such a way that separable and entangled X-states of all
ranks and kinds can be visualized as vertices, sides and interior of a rectangle of side lengths $\mathcal{G}$ and
$\mathcal{H}$. Although $s$ can assume any real value between $0$ and $3/8$, the figure conveys only cases with $s \in\; ]0,1/4]$.
If $s = 0$ then $\mathcal{H} = \mathcal{G} = \mathcal{BC} = 0$, which implies the collapse of the rectangle to the origin of the
parameter space. In this particular case, the resulting states are all rank-1 or rank-2 separable X-states of the forms 
$\ket{i} \! \bra{i} \otimes (\cos^{2}{\vartheta} \ket{0} \! \bra{0} + \sin^{2}{\vartheta} \ket{1} \! \bra{1})$ or $(\cos^{2}{\vartheta}
\ket{0} \! \bra{0} + \sin^{2}{\vartheta} \ket{1} \! \bra{1}) \otimes \ket{i} \! \bra{i}$, for $i \in \{0,1\}$ and $\vartheta \in
[0,\pi/2]$. If $s \in\; ]1/4,3/8]$, it is no longer possible to establish $\mathcal{H} = s$ or $\mathcal{G} = s$, hence no pure
states can occur.
\begin{figure}[h]
\centering
\includegraphics{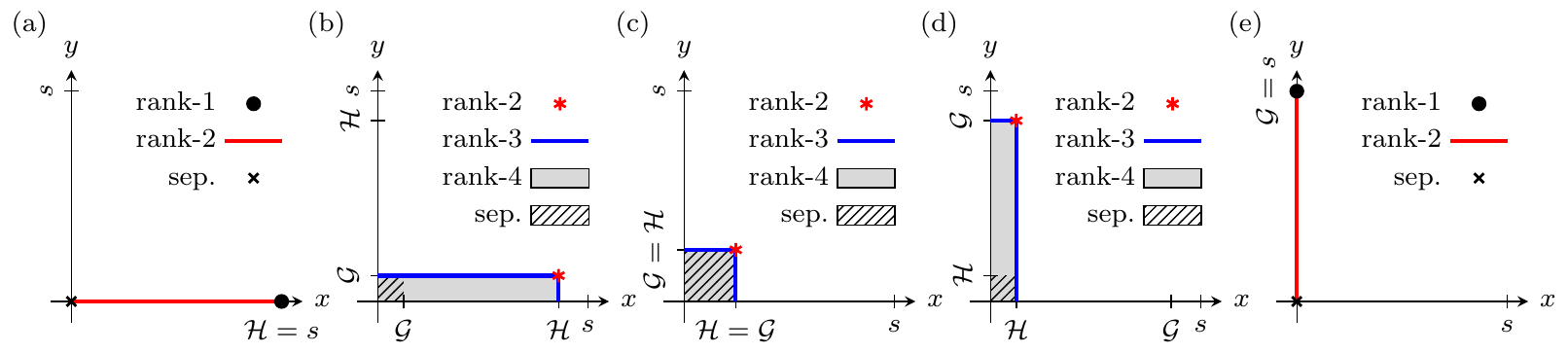}
\caption{Geometric visualization, in the $xy$ parameter space, of classes of X-states with $\mathcal{H} + \mathcal{G} +
\mathcal{BC} = s$, where $s$ is a fixed parameter that sets the scale. The plots are ordered in such a way that $\mathcal{G}$
increases from $0$ to $s$ as one moves from the left to the right. Plot (a) corresponds to a choice of $\theta$, $\varphi$ and
$\psi$ such that $\mathcal{G} = \mathcal{BC} = 0$ and $\mathcal{H} = s$ (e.g., $\theta = \pi/4$, $\varphi = \psi = \pi/2$ and
$s = 1/4$), in such a way that the only possible X-states distribute along the line segment $y = 0$ and $0 \leq x \leq
\mathcal{H}$. Rank-1 X-states of the first kind occupy the end point $(\mathcal{H},0)$, separable rank-2 X-states of the first
kind occupy the origin $(0,0)$, entangled rank-2 X-states of the first kind occupy the points in between. Plot (b) corresponds to
a choice of $\theta$, $\varphi$ and $\psi$ such that $0 < \mathcal{G} < \mathcal{H} < s$ and $\mathcal{BC} > 0$ (e.g., $\theta
\approx 0.598$,  $\varphi \approx 0.962$, $\psi \approx 0.800$ and $s = 1/4$), in such a way that the only possible X-states
distribute along the rectangle of width $\mathcal{H}$ and height $\mathcal{G}$. Rank-2 X-states of the third kind occupy the
vertex $(\mathcal{H},\mathcal{G})$, rank-3 X-states of the first kind occupy the line segment $y = \mathcal{G}$ and $0 \leq x <
\mathcal{H}$, rank-3 X-states of the second kind occupy the line segment $x = \mathcal{H}$ and $0 \leq y < \mathcal{G}$, rank-4 
X-states occupy the interior $0 \leq x < \mathcal{H}$ and $0 \leq y < \mathcal{G}$. States inside (outside) the hatched square
$0 \leq x,y \leq \mathcal{G}$ are separable (entangled). Plot (c) corresponds to a choice of $\theta$, $\varphi$ and $\psi$ such
that $0 < \mathcal{G} = \mathcal{H} < s$ and $\mathcal{BC} > 0$ (e.g., $\theta \approx 0.606$, $\varphi \approx 0.834$, $\psi
\approx 0.436$ and $s = 1/4$), in such a way that the only possible X-states distribute along the square of edge lengths
$\mathcal{H} = \mathcal{G}$. All such states are separable. Plot (d) is analogous to plot (b), but with $\mathcal{G}$ and
$\mathcal{H}$ interchanged. Plot (e) is analogous to plot (a) with $\mathcal{G}$ and $\mathcal{H}$ interchanged. Moreover, in 
(e) both rank-1 and rank-2 states are of the second kind.}\label{fig:Xparameterspace}
\end{figure}

\section{Minimal set of X-states for full occupancy of the two-qubit CP-diagram}\label{sec:minimalset}

As a first application of the parametrizations obtained in the previous section, we now present a construction of a minimal set of
X-states that fully occupy the entangled region of the CP-diagram of generic two-qubit states. We refer to it as a \emph{minimal}
set because its elements are in a one-to-one correspondence with the points of the CP-diagram, in such a way that if a single
state is removed from the set, a point of the CP-diagram is consequently missed. For the reader's convenience, in
\ref{app:entmeasures} we briefly review some basic aspects of the entanglement measure concurrence, and in~\ref{app:cpdiagrams} the boundaries of the two-qubit CP-diagram are explicitly obtained. For more information on two-qubit entanglement versus mixedness diagrams, we refer the reader to Refs.~\cite{01Munro30302,13Hedemann,05Ziman52325,03Wei22110}.

To present our construction we divide the CP-diagram in three disjoint purity domains whose union equals the interval $[1/3,1]$,
where all two-qubit entangled states live \cite{98Zyczkowski883}. For each of these purity subdomains we prove a theorem whose
statement provides parameter values, as functions of the desired purity and concurrence values, that produce a family of X-states
of fixed rank that exhausts the corresponding CP-region. Before stating and proving the theorems, let us briefly outline the
procedure by which the proposed parameter values were obtained. 

An arbitrary X-state $\bm{\varrho}$, parametrized as in (\ref{eq:Xstates_param}), has its purity and concurrence given by the
following formulae:\footnote{While the purity formula follows by direct evaluation of $\tr[\bm{\varrho}^{2}]$ for $\bm{\varrho}$
given by Eq. (\ref{eq:Xstates_param}), the concurrence formula can be easily obtained from a useful specialization, due to Wang
and coworkers \cite{06Wang4343}, of the standard concurrence formula \cite{98Wootters2245} for arbitrary two-qubit states to the
case of two-qubit X-states --- see \ref{app:entmeasures}, in particular Eq. (\ref{eq:YuEberlyConcurrence}).}
\begin{align}
\label{eq:purity}
P(\bm{\varrho}) & = 1 - 2 (\mathcal{BC} + \mathcal{G} - y + \mathcal{H} - x) , \\
\label{eq:concurrence}
C(\bm{\varrho}) & = 2 \max \left[0, \sqrt{x} - \sqrt{\mathcal{G}}, \sqrt{y} - \sqrt{\mathcal{H}} \right] ,
\end{align}
which can be specialized to give the purity and concurrence of X-states of a fixed rank by restricting their parameters according 
to the constraints (\ref{eq:rank1_params}) to (\ref{eq:rank4_params}). From the resulting purity equation for each rank we can
eliminate one of the X-state parameters in favor of $P$ and, hence, rewrite $C$ as a function of $P$ and the remaining X-state
parameters. Then, fixing $C=c$ and $P=p$, with $p$ and $c$ representing any possible values of purity and concurrence for the
specific rank and the relevant purity subdomain, we obtain a transcendental equation that can be solved for the X-state
parameters.
 
Although the constructions presented in the following theorems were obtained by solving such transcendental equations, we refrain
from presenting the constructive steps that led to them. Instead, we state the obtained parameters values in terms of $p$ and $c$
and prove that, for any possible pair $(p,c)$, they: (i) give origin to valid X-density matrices of a given rank and (ii) solve
the equations $P=p$ and $C=c$. 

\begin{theorem}
For every generic rank-1 state of concurrence $c$, a rank-1 X-state of same concurrence can be constructed from equation
\eqref{eq:Xstates_param} by taking
\begin{equation} 
\label{eq:rank1_kind1_parameters}
\theta = \frac{1}{2} \arcsin (c) , \quad \varphi = \psi = \frac{\pi}{2} , \quad x = \frac{c^{2}}{4} \quad \mbox{and} \quad
y = \mu = \nu =0 \, .
\end{equation}
\end{theorem}
\begin{proof} 
We start by showing that for every $c \in [0,1]$ the choice of parameters of \eqref{eq:rank1_kind1_parameters} yields a valid
rank-1 X-state. In fact, using (\ref{eq:rank1_kind1_parameters}) to compute the coefficients $\mathcal{B}$, $\mathcal{G}$ and
$\mathcal{H}$ gives
\begin{equation}
\mathcal{B} = \mathcal{G} = 0 \quad \mbox{and} \quad \mathcal{H} = \frac{c^{2}}{4}
\end{equation}
which complies with $x = \mathcal{H}$, $y = \mathcal{G}$, $\mathcal{B}=0$ and, hence, characterizes the resulting states as 
rank-1 X-states of the first kind. Finally, substituting \eqref{eq:rank1_kind1_parameters} in the concurrence formula
\eqref{eq:concurrence} we obtain $C = c$.
\begin{flushright}
$\Box$
\end{flushright}
\end{proof}

Note that although the choice of parameters \eqref{eq:rank1_kind1_parameters} leads to rank-1 X-states of the first kind, rank-1
X-states of the second kind can also access every $c \in [0,1]$: this is achieved with
\begin{equation}
\theta = \frac{\pi}{2} , \quad \varphi = \frac{1}{2} \arcsin (c) , \quad \psi = x = 0 , \quad y = \frac{c^{2}}{4} \quad \mbox{and}
\quad \mu = \nu =0 \, .
\end{equation}
A proof of this assertion follows the same steps presented above and will be omitted.

\begin{theorem}
\label{thm:rank2}
For every generic state of concurrence $c$ and purity $p \in [5/9,1[$, rank-2 X-states of same concurrence and purity can be
constructed from equation \eqref{eq:Xstates_param} by taking
\begin{equation}
\label{eq:rank2_kind3_parameters}
\theta = \arcsin \left( \sqrt{u} \right) , \quad \varphi = \frac{1}{2} \arcsin \left( \frac{c}{u} \right) , \quad \psi = x = 0 ,
\quad y = \frac{c^{2}}{4} \quad \mbox{and} \quad \mu = \nu = 0 ,
\end{equation}
where
\begin{equation}
\label{eq:defu}
u \coloneq u(p) = \frac{1 + \sqrt{2p-1}}{2} \, .
\end{equation}
\end{theorem}
\begin{proof} 
As shown in~\ref{app:cpdiagrams}, the concurrence of generic two-qubit states with purities $p \in [5/9,1[$ is limited to the interval $c \in [0,u]$, thus $\theta$ and $\varphi$ from Eq. \eqref{eq:rank2_kind3_parameters} are well-defined. Direct computation of
the coefficients $\mathcal{B}$, $\mathcal{G}$ and $\mathcal{H}$ results in
\begin{equation}
\mathcal{B} =u ,\quad \mathcal{G} = \frac{c^{2}}{4} \quad \mbox{and} \quad \mathcal{H} = 0 \, ,
\end{equation}
which complies with $x = \mathcal{H}$, $y = \mathcal{G}$, $\mathcal{BC} > 0$ and, hence, characterizes the resulting states as
rank-2 X-states of the third kind. Finally, note that straightforward evaluation of Eqs. \eqref{eq:purity} and
\eqref{eq:concurrence} with the choice of parameters \eqref{eq:rank2_kind3_parameters} gives, respectively, $P = p$ and $C = c$.
\begin{flushright}
$\Box$
\end{flushright}
\end{proof}

Regarding theorem \ref{thm:rank2}, two remarks are worth pointing out. First, the parameters of Eq.
\eqref{eq:rank2_kind3_parameters} produce valid rank-2 X-states also in the CP-region $p \in [1/2,5/9[$ and $c \in [0,u]$, hence
covering the entire shaded area in Fig. \ref{fig:minset}(a). Since generic \emph{rank-2} states are restricted to the CP-region
$p \in [1/2,1[$ and $c \in [0,u]$ (cf.~\ref{app:cpdiagrams}), we may conclude that the parameters of Eq.~(\ref{eq:rank2_kind3_parameters}) lead to 
X-state counterparts of same concurrence, purity and \emph{rank} for \emph{every} two-qubit state of rank-2. Besides, as we demonstrate
in \ref{app:purrankunit}, any pair of rank-2 density matrices of same purity can be related via unitary conjugation. Hence, the 
X-state counterparts defined by Eq. (\ref{eq:rank2_kind3_parameters}) can be produced by a unitary transformation of a rank-2 two-qubit state 
of the same purity.

Secondly, in theorem \ref{thm:rank2} we relied upon rank-2 X-states of the third kind to exhaust the corresponding CP-region.
Indeed, rank-2 X-states of the first and second kinds cannot achieve concurrences greater than $q \coloneq \sqrt{2p-1}$ 
[dotted line in Fig. \ref{fig:minset}(a)], being thus unsuitable for the task. This can be easily seen by computing the 
purity and concurrence for such states and then combining the resulting expressions to get
\begin{equation}
c = \sqrt{q^{2} - 1 + \sin^{2} (2 \vartheta}) \, ,
\end{equation}
where $\vartheta$ represents $\theta$ (in the case of the first kind parameters) or $\varphi$ (in the case of the second kind
parameters). Clearly, the maximal value of $c$ is $q < u$.

\begin{figure}[h]
\centering
\includegraphics{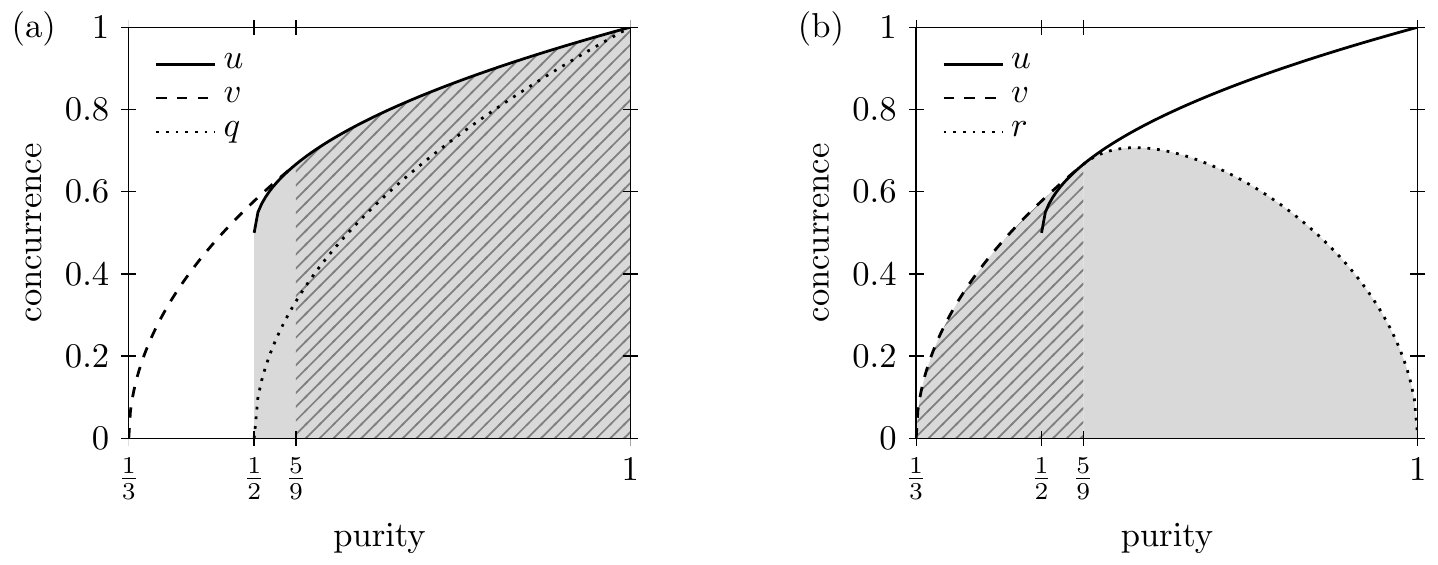}
\caption{Occupation of a two-qubit CP-diagram by rank-2 and rank-3 X-states specified in theorems \ref{thm:rank2} and
\ref{thm:rank3}, where the maximum between solid and dashed lines determines the upper bound of the CP-diagram for generic 
two-qubit states. The shading in (a) represents the accessible region to rank-2 X-states of the third kind [parameters given by
Eq. (\ref{eq:rank2_kind3_parameters})], which coincides with the accessible region for generic rank-2 two-qubit states. Note that
the hatching highlights the region covered by rank-2 elements of our minimal set. Besides, the dotted line gives the upper bound
for rank-2 X-states of the first and second kinds, showing that they can only exhaust the CP-region $p \in [1/2,1[$ and
$c \in [0,q]$. The shading in (b) represents the accessible region to rank-3 X-states of the first kind [parameters given by
Eq. (\ref{eq:rank3_kind1_parameters})], while the hatching highlights the region covered by rank-3 elements of our minimal set 
that completes the non-hatched CP-region in (a). In this plot, the dotted line gives the upper bound for rank-3 X-states of the
first kind over the extended domain $[5/9,1[$, showing that \emph{these} rank-3 X-states do not exhaust the entire CP-region.}
\label{fig:minset}
\end{figure}

\begin{theorem}
\label{thm:rank3}
For every generic state of concurrence $c$ and purity $p \in [1/3,5/9[$, rank-3 X-states of same concurrence and purity can be
constructed from equation \eqref{eq:Xstates_param} by taking
\begin{equation}
\label{eq:rank3_kind1_parameters}
\theta = \arcsin \left( \sqrt{2w} \right) , \quad \varphi = \frac{\pi}{4} , \quad \psi = \frac{\pi}{2} , \quad 
x = \frac{c^{2}}{4} \quad \mbox{and} \quad y = \mu = \nu = 0 \, ,
\end{equation}
where
\begin{equation}
\label{eq:defw}
w \coloneq w(v(p),c) = \frac{1}{3} - \frac{1}{2} \sqrt{\frac{v^{2}}{3} - \frac{c^{2}}{3}} \quad \mbox{and} \quad 
v \coloneq v(p) = \sqrt{2p - \frac{2}{3}} \, .
\end{equation}
\end{theorem}
\begin{proof} 
As shown in~\ref{app:cpdiagrams}, the concurrence of generic two-qubit states with purities $p\in[1/3,5/9[$ is limited to the interval $c \in [0,v]$. In this case, some analysis of Eq. (\ref{eq:defw}) reveals that $w \in \left[ \frac{1}{3} - \frac{1}{3 \sqrt{3}},
\frac{1}{3} \right]$, rendering $\theta$ from Eq. \eqref{eq:rank3_kind1_parameters} well-defined. Direct computation of the
coefficients $\mathcal{B}$, $\mathcal{G}$ and $\mathcal{H}$ gives
\begin{equation}
\mathcal{B} = w , \quad \mathcal{G} = 0 \quad \mbox{and} \quad \mathcal{H} = \frac{1}{3} \left( 1 - p + \frac{c^{2}}{2} - w
\right) \, ,
\end{equation}
which complies with $x < \mathcal{H}$, $y = \mathcal{G}$, $\mathcal{B} > 0$ and, hence, characterizes the resulting states as
rank-3 X-states of the first kind. To see that $x < \mathcal{H}$, note the following:
\begin{equation}
x = \frac{c^{2}}{4} \leq \frac{1}{2} \left( p - \frac{1}{3} \right) < \frac{2}{3} - p < \frac{1}{3} \left( \frac{2}{3} - p
\right) + \frac{2}{3} x = \frac{1}{3} \left( 1 - p + \frac{c^{2}}{2} - \frac{1}{3} \right) \leq \mathcal{H} \, ,
\end{equation}
where the first, second and fourth inequalities follow from the upper bounds for $c$, $p$ and $w$, respectively. This particular
choice of parameters allows, through the straightforward evaluation of Eqs. \eqref{eq:purity} and \eqref{eq:concurrence}, to 
obtain $P = p$ and $C = c$. 
\begin{flushright}
$\Box$
\end{flushright}
\end{proof}

A few remarks about theorem \ref{thm:rank3} are due. First, by imposing $\mathcal{B} > 0$ and $x < \mathcal{H}$ to the choice of
parameters of Eq. (\ref{eq:rank3_kind1_parameters}), we find that they also produce valid rank-3 X-states in the CP-region
$p \in [5/9,1[$ and $c \in [0,r[$, where
\begin{equation}
r \coloneq r(p) = \sqrt{2} \sqrt{1 - 2p + \sqrt{2p-1}} \, .
\end{equation}
However, since in this purity range the concurrence of generic two-qubit states goes up to $u \geq r$, such a choice does not fill
the entire CP-region in the purity interval $[5/9,1[$, as shown with shading in Fig. \ref{fig:minset}(b). Although a different
choice of rank-3 parameters could be tailored to exhaust that region, for now we shall leave its occupancy for the choice of
rank-2 parameters of Eq. (\ref{eq:rank2_kind3_parameters}), as shown with hatching in Fig. \ref{fig:minset}(a).

Second, although the choice of parameters of Eq. (\ref{eq:rank3_kind1_parameters}) yields rank-3 X-states of the first kind, it is
also possible to access every $p \in [1/2,5/9[$ and $c \in [0,v]$ with rank-3 X-states of the second kind. This can be achieved,
for example, with
\begin{equation}
\label{eq:rank3_kind2_parameters}
\theta = \arcsin (\sqrt{z}) \, , \quad \varphi = \frac{\pi}{4} , \quad \psi = x = 0 , \quad y = \frac{c^{2}}{4} \quad \mbox{and}
\quad \mu = \nu = 0 \, ,
\end{equation}
where
\begin{equation}
z \coloneq \left\{ \begin{array}{cc}
\frac{4}{3} - 2w & \mbox{if} \quad 2p \leq 1+c^{2} \\
2w & \mbox{if} \quad 2p > 1+c^{2}
\end{array} \right. \, .
\end{equation}
A proof of this follows the same steps presented above and will be omitted. We only note that also this choice of parameters can
be extended to the purity domain $[5/9,1[$. However, this is only possible for certain values of $c$ which (i) do not cover every
point already visited by the extension of the parameters of Eq. (\ref{eq:rank3_kind1_parameters}) [shaded and non-hatched region
in Fig. \ref{fig:minset}(b)] and (ii) do not cover every point left unvisited by the extension of the parameters of Eq.
(\ref{eq:rank3_kind1_parameters}) [empty area between the solid and dotted lines in Fig. \ref{fig:minset}(b)].

Third, although for every rank-3 two-qubit state of purity $p\in[1/3,5/9[$ we have constructed a X-state of same purity, rank and
concurrence, that does not mean that our construction is related to the input state via a unitary conjugation. As we demonstrate
in \ref{app:purrankunit}, such a conclusion can only be drawn in the case of rank-1 and rank-2 states. For example, consider the
two-qubit density matrix of rank-3
\begin{equation}
\frac{1}{40} \left[ \begin{array}{cccc}
13 & 3 \sqrt{3} & 2 \sqrt{3} & -10 \\
3 \sqrt{3} & 7 & 6 & -2 \sqrt{3} \\
2 \sqrt{3} & 6 & 7 & -3 \sqrt{3} \\
-10 & -2 \sqrt{3} & -3 \sqrt{3} & 13
\end{array} \right] \, ,
\end{equation}
which has $p=0.54$ and $c=0.4$. The following matrices are rank-3 X-state counterparts constructed according to the parameter values of
Eqs. (\ref{eq:rank3_kind1_parameters}) and (\ref{eq:rank3_kind2_parameters}), respectively:
\begin{equation}
\frac{1}{30} \left[ \begin{array}{cccc}
10 + 2 \sqrt{19} & \cdot & \cdot & 6 \\
\cdot & 10 - \sqrt{19} & \cdot & \cdot \\
\cdot & \cdot & \cdot & \cdot \\
6 & \cdot & \cdot & 10 - \sqrt{19}
\end{array} \right]
\quad \mbox{and} \quad
\frac{1}{30} \left[ \begin{array}{cccc}
10 - 2 \sqrt{19} & \cdot & \cdot & \cdot \\
\cdot & 10 + \sqrt{19} & 6 & \cdot \\
\cdot & 6 & 10 + \sqrt{19} & \cdot \\
\cdot & \cdot & \cdot & \cdot
\end{array} \right] \, .
\end{equation}
Although the three matrices share the same rank, purity and concurrence, each one displays a different set of eigenvalues, being
thus impossible to be related via unitary conjugation. Of course, this does not preclude the existence of yet another 
X-state counterpart that could be obtained via unitary conjugation of the input density matrix. The existence of such counterparts will
be proved in the next section.

As a summary of the main results of this section, we now explicitly state the matrix forms (in the computational basis) of the
elements of our minimal set $\bm{\mathcal{S}}_{X}$ -- formed from the parameter choices of Eqs. (\ref{eq:rank1_kind1_parameters}),
(\ref{eq:rank2_kind3_parameters}) and (\ref{eq:rank3_kind1_parameters}) -- namely,
\begin{equation}
\bm{\mathcal{S}}_{X} = \left\{ \bm{\varrho}_{i}(p,c) \, \left| \, i = 1\; \mbox{for}\; p = 1 ,\, i = 2\;\mbox{for}\; p \in \left[
\frac{5}{9},1 \right[ ,\, i = 3\;\mbox{for}\; p \in \left[ \frac{1}{3},\frac{5}{9} \right[ \right. \right\} \, ,
\end{equation}
with
\begin{align}
\bm{\varrho}_{1}(c) & = \frac{1}{2} \left[ \begin{array}{cccc}
1 + \sqrt{1-c^{2}} & \cdot & \cdot & c \\
\cdot & \cdot & \cdot & \cdot \\
\cdot & \cdot & \cdot & \cdot \\
c & \cdot & \cdot & 1 - \sqrt{1-c^{2}}
\end{array} \right] \, , \\
\bm{\varrho}_{2}(u(p),c) & = \frac{1}{2} \left[ \begin{array}{cccc}
2 - 2u & \cdot & \cdot & \cdot \\
\cdot & u + \sqrt{u^{2} - c^{2}} & c & \cdot \\
\cdot & c & u - \sqrt{u^{2}-c^{2}} & \cdot \\
\cdot & \cdot & \cdot & \cdot
\end{array} \right] \, , \\
\bm{\varrho}_{3}(w(p,c),c) & = \frac{1}{2} \left[ \begin{array}{cccc}
2 - 4w & \cdot & \cdot & c \\
\cdot & 2w & \cdot & \cdot \\
\cdot & \cdot & \cdot & \cdot \\
c & \cdot & \cdot & 2w 
\end{array} \right] \, ,
\end{align}
where $u=u(p)$ and $w=w(p,c)$ were defined in Eqs. (\ref{eq:defu}) and (\ref{eq:defw}), respectively. We remark that although many
other minimal sets that exhaust the entangled region of the CP-diagram of two-qubit states do exist, $\bm{\mathcal{S}}_{X}$ has
the advantage of being highly sparse and formed exclusively by rank-deficient X-states.

Let us conclude this section with a word of caution: the exhaustion of the CP-diagram with elements of $\bm{\mathcal{S}}_{X}$ does
not imply that other entanglement versus mixedness diagrams will also be exhausted by $\bm{\mathcal{S}}_{X}$. This is illustrated
in Fig. \ref{fig:negvspurity} in the case of the negativity versus purity diagram (cf. \ref{app:entmeasures} for the definition
and a brief review of the entanglement measure negativity). In this figure, the thick line bounds the negativity of generic 
two-qubits states of a fixed purity (for more details, see Ref. \cite{03Wei22110}). The shading highlights the accessible region
to the elements of $\bm{\mathcal{S}}_{X}$ and it was obtained by numerically computing their negativity. Noticeably, the shading
does not completely fill the area below the thick line. This is better understood by recalling that different entanglement
measures quantify different ``types of entanglement'' \cite{09Horodecki865}. So, while the states of $\bm{\mathcal{S}}_{X}$
exhaust the possible values of concurrence-like-entanglement for a fixed purity, they may (and do) lack some values of 
negativity-like-entanglement. 

The magnification glass in Fig. \ref{fig:negvspurity} also shows that, for $p<5/9$, the shading goes beyond the (thin) line
generated by the elements of $\bm{\mathcal{S}}_{X}$ with maximal concurrence per purity (such states formed the border of the 
CP-diagram in Fig. \ref{fig:minset}). This is also due to the existence of many types of entanglement and, in particular, to 
the fact that any two different entanglement measures place different orderings on the set of density matrices
\cite{99Eisert145,00Virmani31}: although the states that generate points above the line are obviously less concurrence-entangled
than the maximally concurrence-entangled states of same purity, they are more negativity-entangled than the latter. Curiously,
though, for $p > 5/9$, the elements of $\bm{\mathcal{S}}_{X}$ with maximal concurrence are also the elements of highest
negativity.
\begin{figure}[h]
\centering
\includegraphics{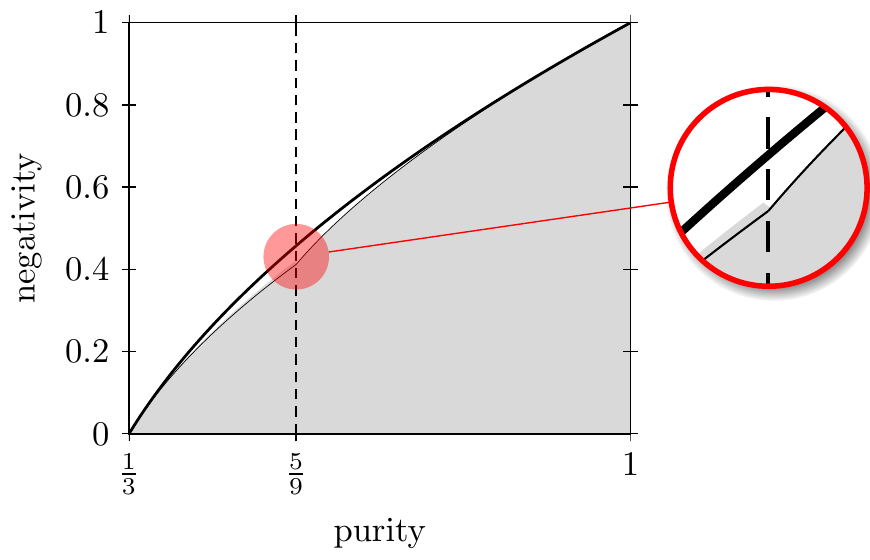}
\caption{Occupation of the negativity versus purity diagram by the elements of $\bm{\mathcal{S}}_{X}$. The thick line gives the
boundary of the diagram for generic two-qubit states and the shading highlights the region occupied by the elements of
$\bm{\mathcal{S}}_{X}$. The thin line indicates the negativity of the states of $\bm{\mathcal{S}}_{X}$ of maximal concurrence 
for a fixed purity. The magnification glass shows that for $p < 5/9$ there are elements of $\bm{\mathcal{S}}_{X}$ that generate
points beyond the thin line, evidencing different orderings on the set of density matrices imposed by negativity and concurrence.}
\label{fig:negvspurity}
\end{figure}

\section{Entanglement Universality via Unitary Evolution}\label{sec:universality}

So far, we have established a weak form of X-state entanglement universality: there are more than necessary X-states to visit every point of a generic two-qubit CP-diagram. As we have seen, though, our produced X-state counterparts cannot, in general, be obtained from the input states via a unitary transformation, nor they will exhaust entanglement versus purity diagrams other than the one in which entanglement is quantified by concurrence. Needless to say, however, is that the possibility of coherently producing X-state counterparts to achieve any value of entanglement, as measured by any entanglement measure, would be very interesting from both fundamental and practical viewpoints.

In this section we considerably strengthen our preliminary universality result to accommodate some of the aforementioned desiderata. Specifically, we claim to be always possible to coherently produce X-state counterparts for any two-qubit state preserving (not simultaneously) its concurrence, negativity or relative entropy of entanglement (cf. \ref{app:entmeasures} for a brief review of these entanglement measures). Throughout, we shall refer to these coherently produced X-state counterparts as \emph{X-counterparts}. Of course, the requirement of coherent preparation implies in preservation of mixedness (e.g. purity or von Neumann entropy), so that X-counterparts will exhaust many types of entanglement versus mixedness diagrams (e.g., all types considered in Ref.~\cite{03Wei22110}).

An important step toward proving this stronger universality claim was given by Verstraete \emph{et al.}~\cite{01Verstraete12316}, who showed that concurrence, negativity and relative entropy of entanglement of a generic two-qubit state $\bm{\rho}_G^{ent}$ is maximized by conjugation with a unitary matrix of the form $\bm{\mathcal{U}}=(\bm{U}_1\otimes \bm{U}_2) \bm{\mathcal{O}} \bm{D}_\phi \bm{\Phi}^\dagger$, where $\bm{U}_1$ and $\bm{U}_2$ are arbitrary local unitary transformations, $\bm{D}_\phi$ is a unitary diagonal matrix, $\bm{\Phi}$ is the unitary matrix such that $\bm{\Phi}^\dagger \bm{\rho}_G^{ent} \bm{\Phi}$ is the diagonal matrix of eigenvalues of $\bm{\rho}_G^{ent}$ sorted in non-ascending order, and $\bm{\mathcal{O}}$ is the improper orthogonal matrix
\begin{equation}
\bm{\mathcal{O}}=\left[
\begin{array}{cccc}
\cdot & \cdot & \cdot & 1\\
\frac{1}{\sqrt{2}} & \cdot & \frac{1}{\sqrt{2}} & \cdot\\[2mm]
\frac{1}{\sqrt{2}} & \cdot & -\frac{1}{\sqrt{2}} & \cdot\\
\cdot & 1 & \cdot & \cdot
\end{array}
\right]\,.
\end{equation}
From here, it is immediate to find that (up to local unitary transformations), the density matrix of eigenvalues $\lambda_1\geq\lambda_2\geq\lambda_3\geq\lambda_4$ with maximal concurrence, negativity and relative entropy of entanglement is
\begin{equation}\label{eq:MEMS}
\bm{\rho}_{X}^{max}=\bm{\mathcal{U}}\bm{\rho}_G^{ent}\bm{\mathcal{U}}^\dagger=\frac{1}{2}\left[\begin{array}{cccc}
2\lambda_4 & \cdot & \cdot & \cdot \\
\cdot & \lambda_1+\lambda_3 & \lambda_1-\lambda_3 & \cdot\\
\cdot & \lambda_1-\lambda_3 & \lambda_1+\lambda_3 & \cdot\\
\cdot & \cdot & \cdot & 2\lambda_2
\end{array}\right]\,.
\end{equation}
Thus, for a fixed spectrum, the maximally entangled mixed state (under the three considered measures) is a X-state. 

It follows from this observation and from the intermediate value theorem (see, e.g., Ref.~\cite{67Apostol}), that to prove our stronger universality claim it suffices to show that (i) any entangled X-state can be disentangled via a unitary transformation that preserves the X-form, and (ii) concurrence, negativity and relative entropy of entanglement vary continuously during the referred disentangling evolution. In fact,  if (i) and (ii) are true, then the X-counterpart of $\bm{\rho}_G^{ent}$ can be prepared by composing two unitary evolutions, $\bm{\mathcal{U}}$ and $\bm{\mathcal{V}}$, as in
\begin{equation}
\bm{\rho}_G^{ent} \stackrel{\bm{\mathcal{U}}}{\longrightarrow} \bm{\rho}_{X}^{max} \stackrel{\bm{\mathcal{V}}}{\longrightarrow} \bm{\rho}_{X}^{ent}
\end{equation}
where $\bm{\mathcal{V}}$ denotes a unitary transformation that initiates a X-form preserving and disentangling transformation of $\bm{\rho}_X^{max}$, which is aborted when the instantaneous X-state reaches either the concurrence, negativity or relative entropy of entanglement of the initial state $\bm{\rho}_G^{ent}$ .

The remainder of this section is devoted to prove assertions (i) and (ii).

\subsection{Coherent Disentanglement with X-form Preservation}\label{sec:disentanglement}

We start by considering the unitary transformation induced by the following unitary matrix with $b_k \in [0,2\pi]$ and  $k=1,\ldots,4$:

\begin{equation}\label{eq:Vb}
\bm{V}=\left[\begin{array}{cccc}
\cos{b_1} & \cdot & \cdot & {\rm e}^{i b_2}\sin{b_1}\\
\cdot & \cos{b_3} & {\rm e}^{i b_4} \sin{b_3} & \cdot\\
\cdot & -{\rm e}^{-i b_4} \sin{b_3} & \cos{b_3} & \cdot\\
-{\rm e}^{-i b_2} \sin{b_1} & \cdot & \cdot & \cos{b_1} 
\end{array}
\right]\,.
\end{equation}
Clearly, $\bm{V}$ consists of two independent $SU(2)$ elements applied to the subspaces spanned by $\{\ket{00}, \ket{11}\}$ and $\{\ket{01}, \ket{10}\}$. Since a X-state can be seen as two fictitious qubits living in each of these subspaces, conjugation of an arbitrary X-state with $\bm{V}$ will necessarily preserve the X-form.\footnote{Note, however, that $\bm{V}$ does not induce the \emph{most general} unitary transformation that preserves the X-form. First, the most general element of $SU(2)$ has $3$ parameters (disconsidering an unimportant global phase), whereas each $SU(2)$ element in~(\ref{eq:Vb}) has only $2$ parameters. Second, even if we employed the most general $SU(2)$ parametrization, it is not difficult to see that it is possible to preserve the X-form with unitary transformations that are not of the X-form. Two obvious examples are the unitary transformations induced by conjugation with the unitary matrices $\Id_2\otimes\bm{\sigma}_1$ and $\bm{\sigma}_1\otimes\Id_2$.}

For what follows, it will prove itself useful to determine how certain parameters of an arbitrary X-state change under the unitary transformation induced by $\bm{V}$. Let $\bm{\rho}$ and $\bm{\rho}^\prime=\bm{V}\bm{\rho} \bm{V}^\dagger$ be two X-states with parameters $\{\theta,\varphi,\psi,x,y,\mu,\nu\}$ and $\{\theta^\prime,\varphi^\prime,\psi^\prime,x^\prime,y^\prime,\mu^\prime,\nu^\prime\}$, defined according to Eq.~(\ref{eq:Xstates_param}). Then, some straightforward (however tedious) computation gives

\begin{equation}\label{eq:xyprime}
\begin{array}{rcl}
x^\prime&=&h^2\sin^2{b_1}\cos^2{b_1}-h\sqrt{x} \sin{(2b_1)} \cos{(2b_1)}\cos\delta+x\left[1-\sin^2{(2b_1)}\cos^2\delta\right]\,,\\
y^\prime&=&g^2\sin^2{b_3}\cos^2{b_3}-g\sqrt{y} \sin{(2b_3)} \cos{(2b_3)}\cos\Delta+y\left[1-\sin^2{(2b_3)}\cos^2\Delta\right]\,,
\end{array}
\end{equation}
where, for brevity, we have defined $\delta \coloneq b_2-\mu$, $\Delta \coloneq b_4-\nu$,
\begin{equation}\label{eq:gh}
h\coloneq h(\theta,\varphi,\psi)= \cos^2\theta-\sin^2\theta \sin^2\varphi \sin^2\psi\quad\mbox{and}\quad g\coloneq g(\theta,\varphi,\psi)= \sin^2\theta(\cos^2\varphi-\sin^2\varphi \cos^2\psi)\,.
\end{equation}
Moreover, the following conservation laws can be easily established from the invariance of the trace and the determinant of a matrix unitarily conjugated:
\begin{align}
\mathcal{C}^\prime = \mathcal{C}\quad&\mbox{and}\quad \mathcal{B}^\prime = \mathcal{B}\,,\label{eq:conslaw1}\\
\mathcal{G}^\prime-y^\prime=\mathcal{G}-y\quad&\mbox{and}\quad\mathcal{H}^\prime-x^\prime=\mathcal{H}-x\,,\label{eq:conslaw2}
\end{align}
where $\mathcal{B}^{(\prime)}$, $\mathcal{C}^{(\prime)}$, $\mathcal{G}^{(\prime)}$ and $\mathcal{H}^{(\prime)}$ were defined in Eq.~(\ref{eq:calBCGH}). 

Let us now see how to set the parameters of $\bm{V}$ in order to turn it into a disentangling unitary transformation for any entangled X-state. From Sec. \ref{sec:separable}, we know that $\bm{\rho}^\prime=\bm{V}\bm{\rho} \bm{V}^\dagger$ will be a separable density matrix if and only if $x^\prime$ and $y^\prime$ are no greater than the minimum between $\mathcal{G}^\prime$ and $\mathcal{H}^\prime$, cf. Eq.~(\ref{eq:sep_params}). Combined with the conservation law~(\ref{eq:conslaw2}), this condition can be rewritten as
\begin{equation}\label{eq:sepcondprime}
x^\prime\leq\min[\mathcal{G}-y+y^\prime,\mathcal{H}-x+x^\prime]\qquad\mbox{and}\qquad y^\prime\leq\min[\mathcal{G}-y+y^\prime,\mathcal{H}-x+x^\prime]\,.
\end{equation}
From a strictly algebraic viewpoint, a simple choice of $x^\prime$ and $y^\prime$ that fulfills both inequalities immediately comes out. Consider, first, the case where the input state $\bm{\rho}$ has $\mathcal{H}>\mathcal{G}$, being thus identified with  a point in the $xy$ parameter space of Fig.~\ref{fig:disentangle}(a). If conjugation with $\bm{V}$ can move that point to the left in order to make $x^\prime=\mathcal{G}$, while keeping its ordinate constant, i.e. $y^\prime=y$, then the inequalities~(\ref{eq:sepcondprime}) become
\begin{equation}
\mathcal{G}\leq\min[\mathcal{G},\mathcal{G}+\mathcal{H}-x]\qquad\mbox{and}\qquad y\leq\min[\mathcal{G},\mathcal{G}+\mathcal{H}-x]\,.
\end{equation}
Noticeably, the first inequality is satisfied with saturation, as the minimization yields $\mathcal{G}$ thanks to the positive semidefiniteness of $\bm{\rho}$ that requires $x\leq \mathcal{H}$ [cf. Eq.~(\ref{eq:rangesxy})]. For the same reason, the second inequality becomes $y\leq\mathcal{G}$, whose validity also follows from the positive semidefiniteness of $\bm{\rho}$.

Analogously, if the input state $\bm{\rho}$ has $\mathcal{G}>\mathcal{H}$, as illustrated in Fig.~\ref{fig:disentangle}(b), moving it down such that $x^\prime=x$ and $y^\prime=\mathcal{H}$, turns inequalities~(\ref{eq:sepcondprime}) into
\begin{equation}\label{eq:xyprimeGgeqH}
x\leq\min[\mathcal{G}-y+\mathcal{H},\mathcal{H}]\qquad\mbox{and}\qquad \mathcal{H}\leq\min[\mathcal{G}-y+\mathcal{H},\mathcal{H}]\,.
\end{equation}
In this case, the result of the minimization is $\mathcal{H}$, from which follows that~(\ref{eq:xyprimeGgeqH}) reduces to the always-true inequalities $x\leq\mathcal{H}$ and $\mathcal{H}\leq\mathcal{H}$.

\begin{figure}[h]
\centering
\includegraphics{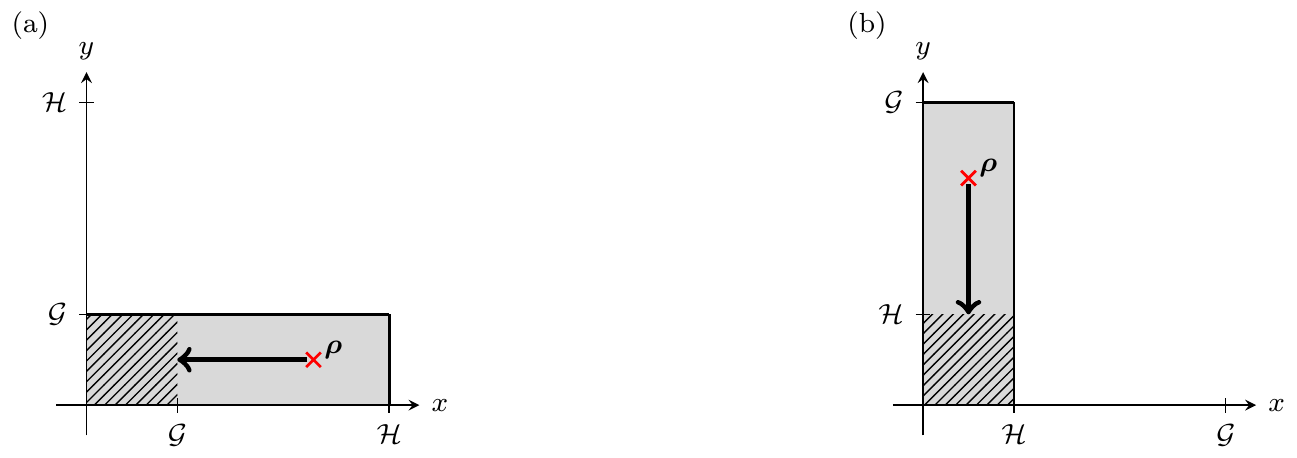}
\caption{Proposed disentangling protocol for an arbitrary two-qubit X-state $\bm{\rho}$. In (a), $\bm{\rho}$ is supposed to have $\mathcal{H}>\mathcal{G}$, in which case a disentangling transformation, represented by the left arrow, would produce a X-state $\bm{\rho}^\prime$ with $x^\prime=\mathcal{G}$ and $y^\prime=y$. In (b), $\bm{\rho}$ is supposed to have $\mathcal{G}>\mathcal{H}$. In this case, the analogous disentangling transformation (down arrow) outputs $\bm{\rho}^\prime$ with $x^\prime=x$ and $y^\prime=\mathcal{H}$.
}\label{fig:disentangle}
\end{figure}

In a nutshell, algebraically, inequalities~(\ref{eq:sepcondprime}) can be satisfied by choosing $x^\prime$ and $y^\prime$ as follows:\footnote{We need not to consider the case $\mathcal{H}=\mathcal{G}$ since every $\bm{\rho}$ with such property is automatically separable [cf. Fig.~\ref{fig:Xparameterspace}(c)].}

\begin{equation}\label{eq:xyprimedisent}
x^\prime=\left\{\begin{array}{cc}
\mathcal{G} &\mbox{if}\quad \mathcal{H}>\mathcal{G}\\
x &\mbox{if}\quad \mathcal{G}>\mathcal{H}
\end{array}\right.\qquad\mbox{and} \qquad y^\prime=\left\{\begin{array}{cc}
y &\mbox{if}\quad \mathcal{H}>\mathcal{G}\\
\mathcal{H} &\mbox{if}\quad \mathcal{G}>\mathcal{H}
\end{array}\right.\,.
\end{equation}
Before accepting Eq.~(\ref{eq:xyprimedisent}) as a solution, though, we must check whether conjugation of $\bm{\rho}$ with $\bm{V}$ can produce states $\bm{\rho}^\prime$ with such parameters values of $x^\prime$ and $y^\prime$. As it turns out, this is feasible for every entangled X-state $\bm{\rho}$:

\begin{theorem}\label{thm:existenceunitary}
Let $\{\theta, \varphi,\psi,x, y,\mu,\nu\}$ be the set of parameters specifying an arbitrary entangled X-state $\bm{\rho}$, and $\{\mathcal{G}, \mathcal{H}, g, h\}$ the set of associated functions of $\theta$, $\varphi$ and $\psi$ defined in equations~(\ref{eq:calBCGH}) and (\ref{eq:gh}). For every such $\bm{\rho}$, the following choice of parameters for the unitary matrix $\bm{V}$ of equation~(\ref{eq:Vb}) produces a X-state $\bm{\rho}^\prime=V\bm{\rho} V^\dagger$ whose parameters $x^\prime$ and $y^\prime$ are given by equation~(\ref{eq:xyprimedisent}). 

If $\mathcal{H}>\mathcal{G}$, make
\begin{equation}
b_2=\mu\,,\quad b_3=0\,,\quad b_4=\nu
\end{equation}
and $b_1$ such that $\cos(2 b_1)=\pm\sqrt{\mathcal{G}/x}$ if $h=0$. Otherwise ($h\neq 0$), $b_1$ is such that
\begin{equation}\label{eq:solfneq0}
\cos b_1 = \sqrt{\frac{1}{2}+\frac{\tilde{\mathfrak{s}}}{2}\sqrt{1-\mathfrak{z}_-}} \quad\mbox{and}\quad\sin b_1=\sqrt{\frac{1}{2}-\frac{\tilde{\mathfrak{s}}}{2}\sqrt{1-\mathfrak{z}_-}}\,,
\end{equation}
where
\begin{equation}\label{eq:defscaseHgG}
\tilde{\mathfrak{s}}=\sgn[h]\sgn\left[\mathfrak{z}_-\mathcal{X}_-+x-\mathcal{G}\right]\,,\quad \mathfrak{z}_-=\frac{x\mathcal{X}_+ + \mathcal{G}\mathcal{X}_-- |h|\sqrt{x\mathcal{G}(\mathcal{X}_+-\mathcal{G})}}{\mathcal{X}_+^2}\quad\mbox{and}\quad \mathcal{X}_\pm \coloneq \left(\frac{h}{2}\right)^2\pm x\,.
\end{equation}

If $\mathcal{H}<\mathcal{G}$, the expressions for $b_1$ and $b_3$ specified above must be interchanged, and the substitutions $x\rightarrow y$, $h\rightarrow{g}$ and $\mathcal{G}\rightarrow\mathcal{H}$ performed in equation~(\ref{eq:defscaseHgG}).
\end{theorem}

We defer to~\ref{app:proofthm} a proof that the parameters specified in the theorem are well-defined and actually implement the desired transformation.

\subsection{Continuity of Entanglement}\label{sec:continuity}

Now that we have established that an arbitrary two-qubit entangled X-state can be disentangled by conjugation with $\bm{V}$, we shall consider the entanglement dynamics of that state evolving under the action of the strongly continuous one-parameter unitary group $(\bm{V}_\tau)_{\tau\in[0,1]}$, where $\bm{V}_{\tau}\coloneq\exp{(i\bm{H}\tau)}$ and $\bm{H}$ is the (X-formed) Hermitian matrix such that $\exp{(i\bm{H})}=\bm{V}$. The explicit form of $\bm{V}_\tau$ can be promptly obtained from Eq.~(\ref{eq:Vb}) by performing the replacements $b_k\to b_k\tau$ for every $k=1,\ldots,4$ and with $b_k$ chosen according to theorem~\ref{thm:existenceunitary}. 

Clearly, as we vary $\tau$ within the range $[0,1]$, the resulting states $\bm{\rho}_\tau=\bm{V}_\tau \bm{\rho} \bm{V}_\tau^\dagger$ preserve X-form and spectrum, while entanglement varies from the initial value, in $\tau=0$, until zero, in $\tau=1$. In this section, we show that entanglement, as measured by concurrence, negativity and relative entropy of entanglement, varies \emph{continuously} with $\tau$, in such a way that every value of these entanglement measures between zero and the initial value can be reached by suitably choosing $\tau$. Furthermore, equations linking the value of $\tau$ to any desired value of concurrence and negativity are derived. We conduct separate analysis for each entanglement measure.

Before proceeding with the continuity analysis, a few notational points are worth mentioning. Throughout, we add the subindex $\tau$ to every parameter associated with the X-state $\bm{\rho}_\tau$. Accordingly, since $\bm{V}_\tau$ is a matrix of the form~(\ref{eq:Vb}), the conservation laws of Eq.~(\ref{eq:conslaw1}) and~(\ref{eq:conslaw2}) also hold with the primes replaced with the subindex $\tau$, namely:
\begin{align}
\mathcal{C}_\tau = \mathcal{C}\quad&\mbox{and}\quad \mathcal{B}_\tau = \mathcal{B}\,,\label{eq:conslaw1tau}\\
\mathcal{G}_\tau-y_\tau=\mathcal{G}-y\quad&\mbox{and}\quad\mathcal{H}_\tau-x_\tau=\mathcal{H}-x\,.\label{eq:conslaw2tau}
\end{align}
In addition, the parameters $x_\tau$ and $y_\tau$ can be promptly obtained by applying the replacement $b_k\to \tilde{b}_k\tau$ to Eq.~(\ref{eq:xyprime}), where $\tilde{b}_k$ denotes the value of $b_k$ specified in theorem~\ref{thm:existenceunitary}:
\begin{equation}
x_\tau=\left[\frac{h}{2}\sin (2 \tilde{b}_1\tau)-\sqrt{x}\cos (2\tilde{b}_1\tau)\right]^2\qquad\mbox{and}\qquad y_\tau=\left[\frac{g}{2}\sin (2 \tilde{b}_3\tau)-\sqrt{y}\cos (2\tilde{b}_3\tau)\right]^2\,.\label{eq:xytau}
\end{equation}
Identities~(\ref{eq:conslaw1tau}), (\ref{eq:conslaw2tau}) and (\ref{eq:xytau}) will be extensively used in what follows.

\subsubsection{Continuity of Concurrence}\label{sec:contconc}
Using Eqs.~(\ref{eq:YuEberlyConcurrence}) and~(\ref{eq:conslaw2tau}) we can express the concurrence of $\bm{\rho}_\tau$ as the following function of $\tau\in[0,1]$:
\begin{align}
\mathfrak{C}(\tau)&=2\max\left[0,\sqrt{x_\tau}-\sqrt{\mathcal{G}_\tau},\sqrt{y_\tau}-\sqrt{\mathcal{H}_\tau}\right]\nonumber\\
&=2\max\left[0,\sqrt{x_\tau}-\sqrt{\mathcal{G}-(y-y_\tau)},\sqrt{y_\tau}-\sqrt{\mathcal{H}-(x-x_\tau)}\right]\,.\label{eq:maximizationconc}
\end{align}
For definiteness, suppose that the initial entangled X-state $\bm{\rho}$ is such that $\mathcal{H}>\mathcal{G}$. Then, according to Eq.~(\ref{eq:xytau}),
 $y_\tau=y$ and $\mathcal{G}\leq x_\tau\leq x$, in which case the maximization of Eq.~(\ref{eq:maximizationconc}) can be explicitly evaluated by noticing that the third term is never positive\footnote{Note that the third term in the maximization of Eq.~(\ref{eq:maximizationconc}) is a \emph{decreasing} function in $x_\tau$, in such a way that its maximum value is achieved for the minimum value of $x_\tau$. Consequently, the non-positivity of this term is implied by the condition $y\leq \mathcal{H}-(x-\mathcal{G})$, which can be seen to be always satisfied by adding up $x\leq \mathcal{H}$ and $y\leq \mathcal{G}$.} and the second term is never negative. As a result, the concurrence formula simplifies to
\begin{equation}\label{eq:contCHG}
\mathfrak{C}_{\mathcal{H}>\mathcal{G}}(\tau) = 2\left(\sqrt{x_\tau}-\sqrt{\mathcal{G}}\right)\,,
\end{equation}
which reaches the maximum at $x_\tau=x$ (or $\tau=0$), the minimum at $x_\tau=\mathcal{G}$ (or $\tau=1$), and is clearly a continuous real function of $x_\tau$ defined on the interval $[\mathcal{G},x]$. Thanks to the obvious continuity of $x_\tau$ in $\tau$ [cf. Eq.~(\ref{eq:xytau})], it turns immediate that the concurrence is also continuous in $\tau$ defined on the interval $[0,1]$.

A completely analogous argument shows that if $\bm{\rho}$ is such that $\mathcal{G}>\mathcal{H}$, then
\begin{equation}\label{eq:contCGH}
\mathfrak{C}_{\mathcal{G}>\mathcal{H}}(\tau) = 2\left(\sqrt{y_\tau}-\sqrt{\mathcal{H}}\right)\,,
\end{equation}
which is also clearly continuous in $\tau$.

Apart from establishing the continuity of the concurrence in $\tau$, Eqs.~(\ref{eq:contCHG}) and~(\ref{eq:contCGH}) allow the determination of the value of $\tau=\tau_c$ that leads to a X-state of concurrence $c$. For example, if $\bm{\rho}$ is such that $\mathcal{H}>\mathcal{G}$, then the appropriate $\tau_c$ can be found by solving the following transcendental equation for $\tau$:
\begin{equation}
\frac{c}{2}=\left|\frac{h}{2}\sin(2\tilde{b}_1 \tau) - \sqrt{x} \cos(2 \tilde{b}_1 \tau)\right|-\sqrt{\mathcal{G}}\,.
\end{equation}
Numerically, this can be efficiently solved.

\subsubsection{Continuity of Negativity}
This closely follows the steps taken for establishing the continuity of the concurrence. Substitution of Eqs.~(\ref{eq:conslaw1tau}) and~(\ref{eq:conslaw2tau}) into~(\ref{eq:negdefX}) gives the negativity of $\bm{\rho}_\tau$ as the following function of $\tau\in[0,1]$:
\begin{align}
\mathfrak{N}(\tau)&=-\min\left[0,\frac{\mathcal{B}_\tau}{2}-\sqrt{\left(\frac{\mathcal{B}_\tau}{2}\right)^2-\mathcal{G}_\tau+x_\tau},\frac{\mathcal{C}_\tau}{2}-\sqrt{\left(\frac{\mathcal{C}_\tau}{2}\right)^2-\mathcal{H}_\tau+y_\tau}\right]\nonumber\\
&=-\min\left[0,\frac{\mathcal{B}}{2}-\sqrt{\left(\frac{\mathcal{B}}{2}\right)^2-\mathcal{G}+x_\tau+y-y_\tau},\frac{\mathcal{C}}{2}-\sqrt{\left(\frac{\mathcal{C}}{2}\right)^2-\mathcal{H}+y_\tau+x-x_\tau}\right]\,.\label{eq:minimizationneg}
\end{align}

Once again, we first consider an initial entangled X-state with $\mathcal{H}>\mathcal{G}$, for which we have already seen that $y_\tau=y$ and $\mathcal{G}\leq x_\tau\leq x$. Use of these in the minimization of Eq.~(\ref{eq:minimizationneg}) leads to an optimization problem that can be trivially solved by noticing that the third term is never negative\footnote{In this case, the third term in the minimization of Eq.~(\ref{eq:minimizationneg}) is an \emph{increasing} function in $x_\tau$, reaching its minimum value when $x_\tau=\mathcal{G}$. The non-negativity of this term is thus implied by $-\mathcal{H}+y+x-\mathcal{G}\leq 0$, which we have already seen to be true.} and the second term is never positive. The negativity formula thus becomes
\begin{equation}\label{eq:contNHG}
\mathfrak{N}_{\mathcal{H}>\mathcal{G}}(\tau)=-\frac{\mathcal{B}}{2}+\sqrt{\left(\frac{\mathcal{B}}{2}\right)^2+x_\tau-\mathcal{G}}\,,
\end{equation}
which reaches the maximum at $x_\tau=x$ (or $\tau=0$), the minimum at $x_\tau=\mathcal{G}$ and is clearly a continuous real function of $x_\tau$ defined on the interval $[\mathcal{G},x]$. Just as occurred in the analysis of concurrence, the continuity of $x_\tau$ in $\tau$ implies that also the negativity is continuous in $\tau$ defined on the interval $[0,1]$. 

For an initial entangled X-state $\bm{\rho}$ such that $\mathcal{G}>\mathcal{H}$, one can derive the following negativity formula, also obviously continuous in $\tau$:
\begin{equation}\label{eq:contNGH}
\mathfrak{N}_{\mathcal{G}>\mathcal{H}}(\tau)=-\frac{\mathcal{C}}{2}+\sqrt{\left(\frac{\mathcal{C}}{2}\right)^2+y_\tau-\mathcal{H}}\,.
\end{equation}
We conclude by noting that Eqs.~(\ref{eq:contNHG}) and~(\ref{eq:contNGH}) allow the determination of the value of $\tau=\tau_n$ that leads to a X-state of negativity $n$. If $\bm{\rho}$ has $\mathcal{H}>\mathcal{G}$, for example, then $\tau_n$ can be obtained by numerically solving the following transcendental equation for $\tau$:
\begin{equation}
n=-\frac{\mathcal{B}}{2}+\sqrt{\left(\frac{\mathcal{B}}{2}\right)^2+\left[\frac{h}{2}\sin(2\tilde{b}_1 \tau) - \sqrt{x} \cos(2 \tilde{b}_1 \tau)\right]^2-\mathcal{G}}\,.
\end{equation}

\subsubsection{Continuity of Relative Entropy of Entanglement}

Since a closed-form for the relative entropy of entanglement of a two-qubit X-state is still unknown (cf.~\ref{app:entmeasures}), we resort to a powerful continuity property of the relative entropy of entanglement\footnote{An analogous continuity property has been demonstrated in~\cite{00Nielsen64301} for the entanglement of formation. It could have been exploited, in Sec.~\ref{sec:contconc}, to establish the continuity of the concurrence in $\tau\in[0,1]$.} derived by Donald and Horodecki~\cite{99Donald257} elaborating on a celebrated inequality due to Fannes~\cite{73Fannes291}.

The general result of Ref.~\cite{99Donald257} implies that, for any pair of two-qubit density matrices $\bm{\rho}_1$ and $\bm{\rho}_2$ such that $\|\bm{\rho}_1-\bm{\rho}_2\|_{tr}\leq 1/3$, the following Fannes-type inequality holds:
\begin{equation}\label{eq:fannesrelent}
|S_{e}(\bm{\rho}_1)-S_{e}(\bm{\rho}_2)|\leq 8\|\bm{\rho}_1-\bm{\rho}_2\|_{tr} -2\|\bm{\rho}_1-\bm{\rho}_2\|_{tr}\log_2{(\|\bm{\rho}_1-\bm{\rho}_2\|_{tr})}\,,
\end{equation}
where $\|\cdot\|_{tr}$ denotes the trace norm (sum of the singular values).

To see that~(\ref{eq:fannesrelent}) implies that the relative entropy of entanglement is continuous in $\tau\in[0,1]$, let $\tau_0$ and $\delta\tau$ be arbitrary real numbers such that $0<\tau_0+\delta \tau\leq 1$ and, for an arbitrary two-qubit X-density matrix $\bm{\rho}$, define
\begin{equation}\label{eq:rho1rho2tau}
\bm{\rho}_1 = \bm{V}_{\tau_0} \bm{\rho} \bm{V}_{\tau_0}^\dagger\quad\mbox{and}\quad \bm{\rho}_2 = \bm{V}_{\tau_0+\delta\tau} \bm{\rho} \bm{V}_{\tau_0+\delta\tau}^\dagger=\bm{V}_{\delta\tau} \bm{\rho}_1 \bm{V}_{\delta\tau}^\dagger\,.
\end{equation}
The strong continuity property of $(\bm{V}_\tau)_{\tau\in[0,1]}$ ensures that
\begin{equation}
\lim_{\delta\tau\to 0}\bm{\rho}_2 = \bm{\rho}_1\quad\mbox{or, equivalently,}\quad \lim_{\delta \tau\to 0}\|\bm{\rho}_1-\bm{\rho}_2\|_{tr}=0\,,
\end{equation}
so, choosing $\bm{\rho}_1$ and $\bm{\rho}_2$ in inequality~(\ref{eq:fannesrelent}) as in Eq.~(\ref{eq:rho1rho2tau}) and taking the limit $\delta\tau\to0$ on both sides of~(\ref{eq:fannesrelent}) yields
\begin{equation}
\lim_{\delta\tau\to 0}|S_{e}(\bm{\rho}_1)-S_{e}(\bm{\rho}_2)|=0\,,
\end{equation}
which is the desired continuity property.

\section{Concluding Remarks}\label{sec:concludingremarks}

In this paper a form of universality property of the set of two-qubit X-states with respect to two-qubit entanglement was established: for every two-qubit state, we have demonstrated that there exists a corresponding two-qubit X-state of same spectrum and entanglement as measured by three different entanglement measures; concurrence, negativity and relative entropy of entanglement. We followed a constructive approach that culminated with the parametrization of a family of unitary transformations that converts between arbitrary two-qubit states and their corresponding X-counterparts, as well as with a semi-analitical characterization of the parameter values that lead to the conservation of concurrence or negativity. As by-products, we have parametrized the set of density matrices of separable, entangled and rank-specific X-states, highlighting its particularly simple geometry. Besides, we have explicitly constructed a set of X-density matrices whose elements are in a one-to-one correspondence with the points of the CP-diagram for generic two-qubit states; a result that can be considered a weaker version of our main universality result.

Regarding the strength of our main universality result, a relevant point is whether two-qubit X-states would exhibit the same universality property when confronted with arbitrary entanglement monotones. An answer to this is intrinsically related to the long-standing open problem: are the maximally entangled mixed states (of a given spectrum) all the same irrespectively of the entanglement monotone that quantifies their entanglement?~\cite{01Audenaert}. As we have seen, the fact that maximally entangled mixed states are indeed the same for the entanglement monotones concurrence (entanglement of formation), negativity and relative entropy of entanglement~\cite{01Verstraete12316} was crucial in establishing our universality result. In order to extend its generality for \emph{arbitrary} entanglement monotones we would need, at least, to guarantee that every possible maximally entangled mixed state is of the X-form. Of course, even if that proves to be true, one would still need to establish the continuity of every entanglement monotone in the sense of Sec.~\ref{sec:continuity}. Although these are certainly two deep open problems, to the best of our knowledge there are no reasons to discard the possibility of this ultimate universality property of two-qubit X-states.

Another point that draws further attention to the family of two-qubit X-states is motivated by the following question: Is that the \emph{only} family of \emph{sparse} density matrices exhibiting the observed universality property? We give a tentative affirmative answer in the case where the sparsity of the X-states is to be preserved (i.e., at least half of the matrix elements equal to zero). Once again, the argument is based on the fact that, for a given set of eigenvalues, maximally entangled mixed states (with respect to concurrence, negativity and relative entropy of entanglement), have all the same X-form~(\ref{eq:MEMS}) up to local unitary transformations~\cite{01Verstraete12316}. Therefore, any ``candidate shape'' must be attainable by conjugating~(\ref{eq:MEMS}) with local unitaries. Now, an obvious way to \emph{preserve} sparsity is to restrict the set of local unitary transformations to the much smaller set of local \emph{permutations}. As simple inspection shows, of the $24$ elements forming the permutation group of dimension $4$, only $3$ (other than the identity) are \emph{local} unitaries.\footnote{Explicitly:
\begin{equation}
\left[\begin{array}{cccc}
\cdot&1&\cdot&\cdot\\
1&\cdot&\cdot&\cdot\\
\cdot&\cdot&\cdot&1\\
\cdot&\cdot&1&\cdot
\end{array}\right]=\Id_2\otimes\bm{\sigma}_1\,,\quad
\left[\begin{array}{cccc}
\cdot&\cdot&1&\cdot\\
\cdot&\cdot&\cdot&1\\
1&\cdot&\cdot&\cdot\\
\cdot&1&\cdot&\cdot
\end{array}\right]=\bm{\sigma}_1\otimes\Id_2\quad\mbox{and}\quad
\left[\begin{array}{cccc}
\cdot&\cdot&\cdot&1\\
\cdot&\cdot&1&\cdot\\
\cdot&1&\cdot&\cdot\\
1&\cdot&\cdot&\cdot
\end{array}\right]=\bm{\sigma}_1\otimes\bm{\sigma}_1\,.\nonumber
\end{equation}
 } Moreover, conjugation of~(\ref{eq:MEMS}) with any one of these three elements results in another X-density matrix. So, assuming local permutations as the only local unitary and sparsity-preserving transformations, we conclude that there are no other families of two-qubit states -- as sparse as the X-states -- with such a universality property.

We conclude by mentioning two possible avenues of future work. First, it is interesting to ask whether there are higher dimensional generalizations of X-states that feature the same universality property. Along these lines, some preliminary numerical work in~\cite{13Hedemann} suggests that qubit-qutrit X-states already fail to be universal with respect to negativity, which raises the prospect to other classes of sparse states. However, the results of~\cite{13Hedemann} are not conclusive and call for analytical investigation. Also worth of note is the fact that a closed formula for the genuine multipartite concurrence for arbitrary $N$-qubit X-states was recently obtained in~\cite{12Rafsanjani062303} and maximized in~\cite{13Agarwal1350043}, giving rise to the set of $N$-qubit X-states of maximal genuine multipartite concurrence for a given linear entropy (the so-called X-MEMS). Given the formal similarity between the X-state formulas for two-qubit concurrence and $N$-qubit genuine multipartite concurrence, it is conceivable that the universality of $N$-qubit X-states with respect to genuine multipartite concurrence can be established in an analogous way as that presented in this work. However, the success of such an analysis depends on a previous demonstration that the X-MEMS are indeed the maximally entangled mixed states over all $N$-qubit states, a fact that is currently unknown~\cite{private}. Second, the universality of two-qubit X-states with respect to the relative entropy of entanglement implies that a closed formula for the relative entropy of two-qubit X-states can be regarded as a closed-formula for any two-qubit state as long as one is able to relate every two-qubit state with its X-counterpart independently of previous knowledge about the relative entropy of entanglement of the input two-qubit state. Obtaining such a relation and a closed formula for the relative entropy of two-qubit X-states are, thus, two possible steps toward the solution of the long-standing open problem of finding a closed-formula for the relative entropy of entanglement of a two-qubit state~\cite{03Eisert}.

\section*{Acknowledgements}
The authors are grateful to J. Combes for a critical reading of an earlier draft of this manuscript. PEMFM thanks T. M. Stace and G. J. Milburn for discussions and for the offered hospitality at the University of Queensland, where part of this work was carried out. This work is supported by the Brazilian Air Force and by the program ``Ci\^encia sem Fronteiras''.

\appendix
\section{Entanglement Measures}\label{app:entmeasures}
Entanglement quantification is a slippery matter. That is mostly so because, in the general framework of mixed states, we do not fully understand what entanglement is. Instead, we collect examples of tasks that could not be performed with a separable state and, vaguely, identify the amount of entanglement of a non-separable state with its corresponding performances in accomplishing such a task. Even more abstractly, well-behaved bounds on these perfomances are commonly regarded as entanglement measures too. As a result, inequivalent ways of measuring mixed state entanglement abound. 

In contrast, in the limited framework of bipartite pure states, much better understanding of entanglement is available. Consequently, an essentially unique entanglement measure emerges~\cite{96Bennett2046,97Popescu3319}: the entropy of entanglement, defined as the von Neumann entropy of either subsystem considered alone. 

Throughout this paper three different entanglement measures for bipartite \emph{mixed} states are considered, namely: concurrence, relative entropy of entanglement and negativity. In this appendix we briefly review some of their aspects relevant to our purposes. For a much deeper account on the subject of entanglement and its quantification, we refer the reader to \cite{01Horodecki3,09Horodecki865}.

\paragraph{Concurrence}

In~\cite{96Bennett3824}, the problem of determining the amount of Bell states per copy necessary to prepare many copies of a bipartite mixed state $\bm{\rho}$ (resorting only to local operations and classical communications), was taken as a route for quantifying the amount of entanglement of $\bm{\rho}$. The concept of \emph{entanglement of formation} was then materialized as the entanglement measure defined by 
\begin{equation}\label{eq:efdef}
E_f(\bm{\rho})=\min\sum_i p_iE(\ket{\psi_i})\,,
\end{equation}
where $E(\ket{\psi_i})$ is the entropy of entanglement of the pure state $\ket{\psi_i}$, and the minimization is taken over all the pure bipartite state ensembles $\{p_i,\ket{\psi_i}\}$ realizing $\bm{\rho}$.

In a remarkable subsequent development~\cite{98Wootters2245}, Wootters explicitly evaluated the minimization of Eq.~(\ref{eq:efdef}) for $4$-by-$4$ density matrices $\bm{\rho}$, obtaining a closed-formula for the entanglement of formation of two-qubit states:
\begin{equation}\label{eq:woottform}
E_f(\bm{\rho})=\mathfrak{h}\left(\frac{1}{2}\left[1+\sqrt{1-C(\bm{\rho})^2}\right]\right)
\end{equation}
where $\mathfrak{h}(x)\coloneq -x\log_2 x - (1-x)\log_2(1-x)$ and $C(\bm{\rho})$, denominated \emph{concurrence} of $\bm{\rho}$, is given by
\begin{equation}\label{eq:concgen}
C(\bm{\rho})=\max\left[0,\sqrt{\lambda_1}-\sqrt{\lambda_2}-\sqrt{\lambda_3}-\sqrt{\lambda_4}\right]
\end{equation}
with $\lambda_1\geq\lambda_2\geq\lambda_3\geq\lambda_4$ the eigenvalues of $\bm{\rho}(\bm{\sigma}_2\otimes\bm{\sigma}_2)\bm{\rho}^\ast(\bm{\sigma}_2\otimes\bm{\sigma}_2)$, $\bm{\sigma}_2$ being the Pauli-$y$ matrix, and the complex conjugate in $\bm{\rho}^\ast$ taken with respect to $\bm{\rho}$ written in the computational basis.

Inspection of Eq.~(\ref{eq:woottform}) reveals entanglement of formation as a monotonically increasing function of concurrence in the relevant domain $C\in[0,1]$. This simple observation motivated the promotion of $C$ from a mere step in the calculation of $E_f$, to an actual entanglement measure for two-qubit states on its own~\cite{01Wootters27}. 

Remarkable properties of concurrence are: it never increases under local operations or classical communications, remains invariant under local unitary transformations, reduces to the entropy of entanglement when applied to pure states, vanishes if and only if applied to separable states and, quite importantly for our purposes, admits a very simple formula when applied to two-qubit X-states~\cite{06Wang4343}:
\begin{equation}\label{eq:YuEberlyConcurrence}
C(\bm{\rho})=2\max\left[0\,,\,|\rho_{32}|-\sqrt{\rho_{44}\rho_{11}}\,,\,|\rho_{41}|-\sqrt{\rho_{33}\rho_{22}}\right]\,,
\end{equation}
where $\rho_{ij}$ denotes the entry of the X-density matrix $\bm{\rho}$ in the $i$-th row and $j$-th  column.

\paragraph{Relative Entropy of Entanglement}

The notion of quantifying entanglement of an arbitrary quantum state by measuring how distinguishable it is from all separable states of same dimension was introduced in~\cite{97Vedral4452}. Intuitively, the more entangled a state is, the less likely it is to be confused with a separable state after a number of measurements. This idea was materialized in an entanglement measure denominated \emph{relative entropy of entanglement}, for relying upon the quantum relative entropy~\cite{62Umegaki59,00Nielsen} to quantify the ``distance'' between $\bm{\rho}$ and the closest separable state $\bm{\sigma}$, namely:
\begin{equation}\label{eq:relentent}
S_e(\bm{\rho})=\min_{\bm{\sigma}\in \bm{\mathcal{D}}}\tr\left[\bm{\rho}\log\bm{\rho}-\bm{\rho}\log\bm{\sigma}\right]\,,
\end{equation} 
where $\bm{\mathcal{D}}$ represents the set of all separable density operators with the same dimension of $\bm{\rho}$.

Shortly after its inception, Vedral and Plenio showed that the relative entropy of entanglement is a powerful upper bound for yet another entanglement measure~\cite{98Vedral1619}: the \emph{entanglement of distillation}, i.e., the proportion of Bell states that can be distilled from many copies of $\bm{\rho}$ using the optimal purification procedure~\cite{96Bennett722}. In spite of its compelling operational definition, entanglement of distillation proved to be a measure very hard to compute, which reinforced the interest on the bound provided by the relative entropy of entanglement.

Although much more computation-friendly than the entanglement of distillation, evaluation of the relative entropy of entanglement is far from trivial. Numerically, a convergent iterative algorithm~\cite{03Rehacek127904} and an estimator based on semidefinite programming~\cite{10Zinchenko052336} have been proposed, but to date, apart from some cases of high symmetry~\cite{97Vedral2275,98Vedral1619,01Verstraete12316,01Verstraete10327,05Chen2755,02Audenaert032310}, no analytical solution for the minimization problem of Eq.~(\ref{eq:relentent}) has been found; not even in the simplest case of two-qubit density matrices~\cite{03Eisert} (see, however, Refs.~\cite{03Ishizaka060301R,08Miranowicz032310,11Friedland052201}). In the case of two-qubit X-states, the minimization of~(\ref{eq:relentent}) was shown to be equivalent to simultaneously solve a pair of binary non-linear equations whose explicit form depend on the eigenvalues and eigenvectors of $\bm{\rho}$~\cite{10Hui110307}.

Remarkable properties of the relative entropy of entanglement are: it never increases under local operations or classical communications, remains invariant under local unitary transformations, reduces to the entropy of entanglement when applied to pure states, vanishes if and only if applied to separable states, and is also meaningful for the quantification of multipartite entanglement.

\paragraph{Negativity}
Unlike the other measures discussed so far, negativity does not originate from a clear physical context, but rather from the technical desire for an easy-to-compute entanglement measure. It is motivated by the observation, due to Peres~\cite{96Peres1413}, that partial transposition of a separable bipartite density matrix yields another positive semi-definite matrix, which suggests some sort of relationship between entanglement and negative eigenvalues of the partial transpose of a bipartite density matrix. In~\cite{02Vidal032314}, this relationship was materialized in an entanglement measurement defined as the absolute value of the sum of the negative eigenvalues of the partially transposed density matrix or, equivalently,
\begin{equation}\label{eq:negdef}
N(\bm{\rho})=\frac{1}{2}\sum_i|\lambda_i|-\lambda_i=\frac{\|\bm{\rho}^{{\sf T}_A}\|_{tr}-1}{2}\,,
\end{equation}
where $\lambda_i$ are the eigenvalues of a partial transpose of $\bm{\rho}$, here denoted by $\bm{\rho}^{{\sf T}_A}$. The second equation above follows trivially from the normalization of $\bm{\rho}$ and from the definition of the trace norm of a matrix as the sum of its singular values. Of course, in the case of Hermitian matrices such as $\bm{\rho}^{{\sf T}_A}$, the singular values are simply the absolute values of the eigenvalues.

The establishment of Eq.~(\ref{eq:negdef}) as a meaningful entanglement measure is greatly due to the facts that negativity remains invariant under local unitary transformations and never increases under local operations or classical communications~\cite{02Vidal032314}. In addition, a logarithmic variation of negativity provides a useful upper bound for entanglement of distillation. Nevertheless, negativity does not reduce to the entropy of entanglement when applied to pure states, and it can vanish for entangled states because, in higher dimensions, there are entangled states with positive partial transpose~\cite{96Horodecki1,01Horodecki45}.

In the case of two-qubit states, $\bm{\rho}^{{\sf T}_A}$ is known to have at most a single negative eigenvalue~\cite{98Sanpera826}, hence Eq.~(\ref{eq:negdef}) simplifies to
\begin{equation}
N(\bm{\rho})=-\min\left[0,\Lambda_{min}(\bm{\rho}^{{\sf T}_A})\right]\,,
\end{equation}
where $\Lambda_{min}$ denotes the operator that extracts the smallest eigenvalue of its argument. For $\bm{\rho}$ a two-qubit X-state, explicit calculation of the eigenvalues of $\bm{\rho}^{{\sf T}_A}$ yields the following formula for the negativity of two-qubit X-states:
\begin{equation}\label{eq:negdefX}
N(\bm{\rho})=-\min\left[0,\frac{\rho_{22}+\rho_{33}}{2}-\sqrt{\left(\frac{\rho_{22}}{2}\right)^2+\left(\frac{\rho_{33}}{2}\right)^2+|\rho_{41}|^2},\frac{\rho_{11}+\rho_{44}}{2}-\sqrt{\left(\frac{\rho_{11}}{2}\right)^2+\left(\frac{\rho_{44}}{2}\right)^2+|\rho_{32}|^2}\right]\,,
\end{equation}
where $\rho_{ij}$ denotes the entry of the X-density matrix $\bm{\rho}$ in the $i$-th row and $j$-th  column.

\section{Concurrence versus Purity diagram for two-qubit states}\label{app:cpdiagrams}

Concurrence versus purity diagrams provide a simple and visual way to witness how quantum state mixedness limits quantum state entanglement and vice-versa. The CP-diagram for two-qubit states has well known boundaries, meaning that for a given purity $p$, the maximal achievable two-qubit concurrence $c_{max}$ is known. Conversely, for any given concurrence, the smallest reachable two-qubit purity is also known. In this appendix we briefly review some key results that enable the analytical characterization of these boundaries.

Given that the maximally mixed two-qubit state $\Id_4/4$ has purity $p=1/4$, and that pure states have maximal purity ($p=1$), two-qubit CP-diagrams are defined in the purity domain $[1/4,1]$. However, as shown by \.{Z}yczkowski~\cite{98Zyczkowski883}, every two-qubit state with $p\leq 1/3$ is separable, hence $c_{max}(p)=0$ for $p\in[1/4,1/3]$.

For $p>1/3$, though, Munro \emph{et al.}~\cite{01Munro30302} showed that, up to local unitary transformations, two-qubit states with maximal concurrence for a fixed purity $p$ are of the form:
\begin{equation}
\bm{\rho}=\left[
\begin{array}{cccc}
\mathfrak{g}(\gamma) & \cdot & \cdot & \gamma/2\\
\cdot & 1 - 2g(\gamma) & \cdot & \cdot \\
\cdot & \cdot & \cdot & \cdot \\
\gamma/2 & \cdot & \cdot & \mathfrak{g}(\gamma)
\end{array}
\right]\quad\mbox{where}\quad \mathfrak{g}(\gamma)=
\left\{\begin{array}{cc}
\gamma/2 & \mbox{if}\quad \gamma \geq 2/3\\
1/3 & \mbox{if}\quad \gamma < 2/3\\
\end{array}\right.\,,
\end{equation}
from which we can promptly compute $c_{max}=\gamma$ [cf. Eq.~(\ref{eq:YuEberlyConcurrence})] and $p = 1 - 4 \mathfrak{g}(\gamma) + 6 [\mathfrak{g}(\gamma)]^2+\gamma^2/2$. With some simple algebra, we can eliminate the parameter $\gamma$ and write $c_{max}$ as a function of $p$:
\begin{equation}
c_{max}(p)=\left\{\begin{array}{cc}
u(p) &\mbox{if}\quad \frac{5}{9}\leq p \leq 1\\[2mm]
v(p) &\mbox{if}\quad \frac{1}{3}\leq p \leq \frac{5}{9}
\end{array}\right.\,,
\end{equation}
with
\begin{equation}
u(p)=\frac{1}{2}+\frac{1}{2}\sqrt{2p-1}\quad\mbox{and}\quad v(p)=\sqrt{2p-\frac{2}{3}}\,.
\end{equation}

A plot of $c_{max}(p)$ is shown in Fig.~\ref{fig:cpboundaries} and gives the upper boundary of the two-qubit CP-diagram. As it turns out, there are more than sufficient two-qubit states to access every internal point of this diagram. In fact, as shown in Sec.~\ref{sec:minimalset}, there are more than sufficient two-qubit X-states to visit every point $(p,c)$ with $p\in[1/4,1]$ and $c\in[0,c_{max}(p)]$.

\begin{figure}[h]
\centering
\includegraphics{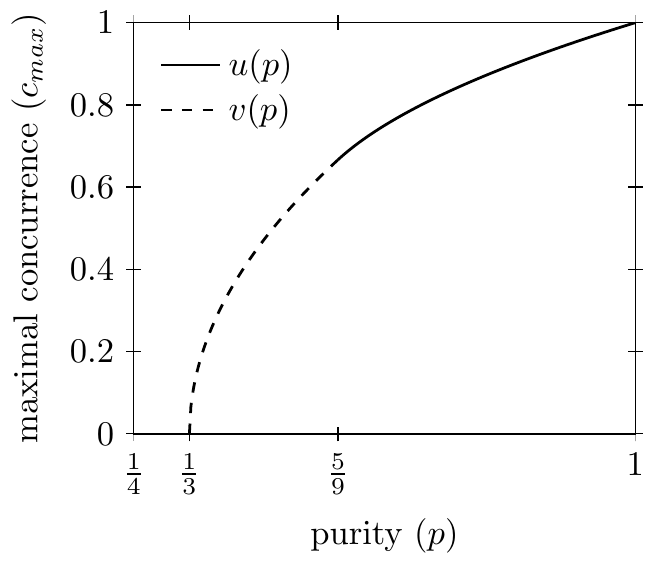}
  \caption{Boundary of the two-qubit CP-diagram.}\label{fig:cpboundaries}
\end{figure}

\section{Density Matrices of same Purity and Rank are not always unitarily related}\label{app:purrankunit}

\begin{lemma}
Let $\bm{\rho}_1$ and $\bm{\rho}_2$ be density matrices of arbitrary dimension $d$. If there exists a unitary transformation $\bm{U}\in SU(d)$ such that $\bm{\rho}_2=\bm{U} \bm{\rho}_1 \bm{U}^\dagger$, then $\tr[\bm{\rho}_2^2]=\tr[\bm{\rho}_1^2]$ and $\rank[\bm{\rho}_2]=\rank[\bm{\rho}_1]$. However, the converse is generally not true, unless $\rank[\bm{\rho}_1]=\rank[\bm{\rho}_2]\leq 2$.
\end{lemma}
\begin{proof}
The direct implication is obvious. For the converse, note that the equivalence between the purities of $\bm{\rho}_1$ and $\bm{\rho}_2$ can be written as
\begin{equation}\label{eq:traceequation}
\tr[\bm{\Lambda}_1^2-\bm{\Lambda}_2^2]=0\,,
\end{equation}
where $\bm{\Lambda}_{i}$ is the matrix of the eigenvalues of $\bm{\rho}_i$ sorted in non-increasing order. Since $\bm{\rho}_i$ are density matrices, we can obtain $\bm{\Lambda}_i$ via a unitary conjugation, namely, $\bm{\rho}_i=\bm{U}_i\bm{\Lambda}_i\bm{U}_i^\dagger$. The key point is this: if Eq.~(\ref{eq:traceequation}) generally implied $\bm{\Lambda}_2=\bm{\Lambda}_1$, then we could write $\bm{U}_2^\dagger\bm{\rho}_2 \bm{U}_2 = \bm{U}_1^\dagger\bm{\rho}_1 \bm{U}_1$ and conclude that $\bm{\rho}_2$ and $\bm{\rho}_1$ are always related by unitary conjugation with $\bm{U}=\bm{U}_2 \bm{U}_1^\dagger$. However, as we now show, the implication
\begin{equation}
\tr[\bm{\Lambda}_1^2-\bm{\Lambda}_2^2]=0 \Rightarrow \bm{\Lambda}_1 = \bm{\Lambda}_2\,,
\end{equation}
only holds if $\bm{\Lambda}_1$ and $\bm{\Lambda}_2$ have at most 2 non-zero entries; that is, $\rank[\bm{\rho}_1]=\rank[\bm{\rho}_2]\leq 2$.

The implication is obviously valid for rank-1. To see that it holds for rank-2 consider the first $2\times 2$ diagonal blocks of $\bm{\Lambda}_1$ and $\bm{\Lambda}_2$, which must be of the form
\begin{equation}
\tilde{\bm{\Lambda}}_1=\left[\begin{array}{cc}
a & \cdot \\
\cdot & 1-a
\end{array}\right]\quad\mbox{and}\quad
\tilde{\bm{\Lambda}}_2=\left[\begin{array}{cc}
p & \cdot \\
\cdot & 1-p
\end{array}\right]\,,
\end{equation}
with $\frac{1}{2}\leq a, p < 1$, where the lower bound follows from the adopted convention of ordering the eigenvalues in $\bm{\Lambda}_i$ in non-increasing order. The remaining entries of $\bm{\Lambda}_1$ and $\bm{\Lambda}_2$ all vanish and can thus be disregarded. It is straightforward to show that 
\begin{equation}
\tr[\bm{\Lambda}_1^2-\bm{\Lambda}_2^2]= (a-p) (a+p-1)\,.
\end{equation}
Given that $a, p \geq \frac{1}{2}$, the expression above will only vanish if $a=p$, or equivalently, $\bm{\Lambda}_1=\bm{\Lambda}_2$, hence proving the claim for the case of rank-2.

Consider now the first $3\times 3$ diagonal blocks of $\bm{\Lambda}_1$ and $\bm{\Lambda}_2$ when $\rank[\bm{\rho}_1]=\rank[\bm{\rho}_2] = 3$,
\begin{equation}
\tilde{\bm{\Lambda}}_1=\frac{1}{2}\left[\begin{array}{ccc}
a+b & \cdot & \cdot\\
\cdot & a-b & \cdot\\
\cdot & \cdot & 2-2a
\end{array}\right]\quad\mbox{and}\quad
\tilde{\bm{\Lambda}}_2=\frac{1}{2}\left[\begin{array}{ccc}
p+q & \cdot & \cdot \\
\cdot & p-q & \cdot \\
\cdot & \cdot & 2-2p
\end{array}\right]\,.
\end{equation}
The adopted convention of non-increasing order of the eigenvalues requires that $0<a+b<2$, $0<a-b\leq a+b$, $0< 2-2a \leq a-b$ and similarly for $p$ and $q$. Some simple analysis shows that these are equivalent to 
\begin{equation}\label{eq:validregion}
\frac{2}{3}\leq a < 1\,,\quad 0\leq b \leq 3a-2\,,\quad \frac{2}{3}\leq p < 1 \quad\mbox{and}\quad 0\leq q \leq 3p-2\,.
\end{equation}
Moreover, in terms of $a$, $b$, $p$ and $q$, the equation $\tr[\bm{\Lambda}_1^2-\bm{\Lambda}_2^2]=0$ takes the form
\begin{equation}
\left(a-\frac{2}{3}\right)^2+\frac{b^2}{3}=\left(p-\frac{2}{3}\right)^2+\frac{q^2}{3}\,,
\end{equation}
which, for any fixed values of $p\neq2/3$ and $q\neq 0$, corresponds to an ellipse in the $ab$ plane, as shown in Fig~\ref{fig:ellipselemma}. Note that for each choice of $(p,q)$ in the figure, there are many points of the ellipse other than $a=p$ and $b=q$ sitting in the shaded region defined by Eq.~(\ref{eq:validregion}). This implies that already in the rank-3 case there are many valid solutions for $\tr[\bm{\Lambda}_1^2-\bm{\Lambda}_2^2]=0$ other than $\bm{\Lambda}_1=\bm{\Lambda}_2$.

\begin{figure}[h]
\centering
\includegraphics{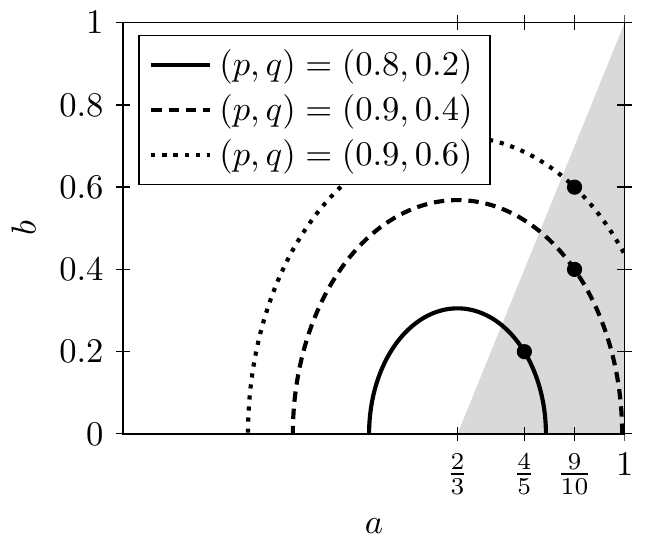}
  \caption{Geometric visualization that two rank-3 density matrices of same purity are generally not related by a unitary transformation. The shading represents the region of the $ab$ parameter space occupied by rank-3 density matrices $\bm{\rho}_1$, whose eigenvalues in non-increasing order are $(a+b)/2$, $(a-b)/2$ and $(1-a)$ [cf. Eq.~(\ref{eq:validregion})]. The points of each arc of ellipse contained in the shading represent the location of rank-3 states $\bm{\rho}_1$ with the same purity of rank-3 (reference) states $\bm{\rho}_{2}$, whose eigenvalues in non-increasing order are $(p+q)/2$, $(p-q)/2$ and $(1-p)$. Of all the points of a given arc, only the highlighted ones correspond to states $\bm{\rho}_1$ unitarily related to $\bm{\rho}_2$.}\label{fig:ellipselemma}
\end{figure}
\begin{flushright}
$\Box$
\end{flushright}
\end{proof}

\section{Proof of Theorem~\ref{thm:existenceunitary}}\label{app:proofthm}
The proof is by construction. Start by making $b_2=\mu$ and $b_4=\nu$ in Eq.~(\ref{eq:xyprime}). After some simple algebra,  we find that
\begin{equation}
x^\prime=\left[\frac{h}{2}\sin (2 b_1)-\sqrt{x}\cos (2b_1)\right]^2\qquad\mbox{and}\qquad y^\prime=\left[\frac{g}{2}\sin (2 b_3)-\sqrt{y}\cos (2b_3)\right]^2\,.\label{eq:xyprime_betas2}
\end{equation}
Clearly, the requirement $x^\prime=x$ (case $\mathcal{G}>\mathcal{H}$) is satisfied by setting $b_1=0$, whereas the requirement $y^\prime=y$ (case $\mathcal{H}>\mathcal{G}$) is satisfied by setting $b_3=0$. Now, only one free parameter of $\bm{V}$ remains in each case. Throughout, we shall consider the case $\mathcal{H}>\mathcal{G}$, where we must find $b_1\in[0,2\pi]$ such that $x^\prime=\mathcal{G}$ for every possible value of $x$, $h$ and $\mathcal{G}$. We note that in the case $\mathcal{G}>\mathcal{H}$, where we must find $b_3\in[0,2\pi]$ such that $y^\prime=\mathcal{H}$ for every possible value of $y$, $g$ and $\mathcal{H}$, a solution can be obtained with the exact same steps to be presented below, but under the replacements: $b_1\rightarrow b_3$, $x\rightarrow y$, $h\rightarrow{g}$ and $\mathcal{G}\rightarrow\mathcal{H}$. 

In order to obtain $b_1$ in the case $\mathcal{H}>\mathcal{G}$, consider first the equation $x^\prime=\mathcal{G}$ in the following form:
\begin{equation}\label{eq:xprimeeqcalG}
\left(\frac{h}{2}\right)^2\sin^2(2b_1) + x\cos^2 (2b_1) - \mathcal{G} = \frac{h}{2}\sqrt{x}\sin (4 b_1)
\end{equation}
If $h=0$, it is immediate that $b_1$ is given by
\begin{equation}
\cos (2b_1)=\pm\sqrt{\frac{\mathcal{G}}{x}}\,.
\end{equation}
Note that the requirement of $\bm{\rho}$ to be entangled implies that $x>\mathcal{G}$, in such a way that $b_1$ is well-defined by the above equation (with any choice of sign) for every entangled X-state with $h=0$ (and $\mathcal{H}>\mathcal{G}$).

In order to solve Eq.~(\ref{eq:xprimeeqcalG}) when $h\neq 0$, we perform the following substitution of variables: 
\begin{equation}\label{eq:chgvar_rtozandsigma}
\sin^2 (2b_1)\equiv \mathfrak{z}\,,\quad \cos^2 (2b_1) = 1-\mathfrak{z}\quad\mbox{and}\quad
\sin (4b_1) = \mathfrak{s}\sqrt{1-(1-2\mathfrak{z})^2}\,,
\end{equation}
where $\mathfrak{s}$ is a binary variable that takes values on the set $\{-1,+1\}$. In terms of $\mathfrak{z}$ and $\mathfrak{s}$, Eq.~(\ref{eq:xprimeeqcalG}) becomes
\begin{equation}\label{eq:eqoriginal22}
\mathfrak{z}\left[\left(\frac{h}{2}\right)^2-x\right]+x-\mathcal{G}=\mathfrak{s}\frac{h}{2}\sqrt{x}\sqrt{1-(1-2\mathfrak{z})^2}\,.
\end{equation}
This is now squared to eliminate $\mathfrak{s}$ and produce an equation for $\mathfrak{z}$ only:
\begin{equation}\label{eq:genquadeqfneq0}
\mathfrak{z}^2\mathcal{X}_+^2-2\mathfrak{z}\left(x\mathcal{X}_++\mathcal{G}\mathcal{X}_-\right)+(x-\mathcal{G})^2=0\,,
\end{equation}
where we have defined $\mathcal{X}_\pm \coloneq \left(\frac{h}{2}\right)^2\pm x$. Note that the condition $\mathcal{X}_+=0$ occurs if and only if $x=h=0$, hence the solution of Eq.~(\ref{eq:genquadeqfneq0}) for any set of X-state's parameters consistent with $h\neq 0$ is given by
\begin{equation}\label{eq:z1pm}
\mathfrak{z}_\pm=\frac{x\mathcal{X}_+ + \mathcal{G}\mathcal{X}_-\pm |h|\sqrt{x\mathcal{G}(\mathcal{X}_+-\mathcal{G})}}{\mathcal{X}_+^2}\,.
\end{equation}
The requirement of entanglement of $\bm{\rho}$ (i.e., $x>\mathcal{G}$) implies that $\mathcal{X}_+>\mathcal{G}$, which, in turn, guarantees that $\mathfrak{z}_\pm\in\mathbb{R}$. Moreover, as we now show, for any set of X-state's parameters consistent with $h\neq0$, we have $\mathfrak{z}_+\in\;]0,1]$ and $\mathfrak{z}_-\in\;]0,1[$. This is established with the standard inequality between the arithmetic and geometric means (AM-GM inequality):
\begin{align}
x\mathcal{X}_++\mathcal{G}\mathcal{X}_-&=\frac{1}{2}\left[2x(\mathcal{X}_+-\mathcal{G})+\frac{\mathcal{G}h^2}{2}\right]\geq |h|\sqrt{x\mathcal{G}(\mathcal{X}_+-\mathcal{G})}\,,\label{ineq:amgm1}\\
\mathcal{X}_+^2-x\mathcal{X}_+-\mathcal{G}\mathcal{X}_-&=\frac{1}{2}\left[\frac{h^2(\mathcal{X}_+-\mathcal{G})}{2}+2x\mathcal{G}\right]\geq |h|\sqrt{x\mathcal{G}(\mathcal{X}_+-\mathcal{G})}\geq -|h|\sqrt{x\mathcal{G}(\mathcal{X}_+-\mathcal{G})}\,.\label{ineq:amgm2}
\end{align}
Clearly, inequality (\ref{ineq:amgm1}) implies that $\mathfrak{z}_-\geq 0$, whereas the first and second inequalities in (\ref{ineq:amgm2}) imply $\mathfrak{z}_+\leq 1$ and $\mathfrak{z}_-\leq 1$, respectively. Since the AM-GM inequality is saturated if and only if the two summands in the arithmetic mean are equal, it is easy to show that $\mathfrak{z}_-=0$ if and only if $x=\mathcal{G}$; that $\mathfrak{z}_+=1$ if and only if $\mathcal{G}=(h/2)^2$; and that $\mathfrak{z}_-=1$ if and only if $\mathcal{G}=(h/2)^2$ and $hx=0$. As $x=\mathcal{G}$ contradicts the hypothesis of $\bm{\rho}$ to be entangled, and $hx=0$ contradicts either $h\neq 0$ or $x>\mathcal{G}$, the only possible saturation is $\mathfrak{z}_+=1$.

So far we have determined and characterized the solutions of Eq.~(\ref{eq:genquadeqfneq0}) for any set of X-state parameters consistent with $h\neq 0$. We now focus on building a solution for the original (pre-squared) equation~(\ref{eq:eqoriginal22}). Naturally, for each set of parameters consistent with $h\neq 0$, the candidate solutions are of the form $\mathfrak{z}=\mathfrak{z}_\pm$ along with some choice of signal $\mathfrak{s}$. As we now show, there is always a choice of signal $\mathfrak{s}=\tilde{\mathfrak{s}}$ that solves Eq.~(\ref{eq:eqoriginal22}) for $\mathfrak{z}=\mathfrak{z}_-$. In order to determine $\tilde{\mathfrak{s}}$, we plug  $\mathfrak{z}=\mathfrak{z}_-$ into Eq.~(\ref{eq:eqoriginal22}) and take the sign function\footnote{The sign function of a real number $\varsigma$ is defined as: 
$
\sgn[\varsigma]\coloneq\left\{\begin{array}{cc}
-1 & \mbox{if } \varsigma<0\,,\\
0 & \mbox{if } \varsigma = 0\,,\\
1 & \mbox{if } \varsigma >0\,.
\end{array}
\right.
$
} on both sides to get
\begin{equation}
\sgn\left[\mathfrak{z}_-\mathcal{X}_-+x-\mathcal{G}\right]=\tilde{\mathfrak{s}}\sgn\left[\frac{h}{2}\sqrt{x}\sqrt{1-(1-2\mathfrak{z}_-)^2}\right]\,.
\end{equation}
Since $x>0$ and $\mathfrak{z}_-\neq 1$ for every valid choice of X-state parameters consistent with $h\neq 0$, the sign function on the rhs can be replaced with $\sgn[h]$. Moreover, the fact that the square of the above identity is satisfied (by construction), guarantees that the sign function on the lhs never vanishes as long as $h\neq 0$. As a consequence, the $\pm 1$ value of $\tilde{\mathfrak{s}}$ is chosen according to
\begin{equation}
\tilde{\mathfrak{s}}=\sgn[h]\sgn\left[\mathfrak{z}_-\mathcal{X}_-+x-\mathcal{G}\right]\,.
\end{equation}
Finally, substituting $\mathfrak{z}=\mathfrak{z}_-$ and $\mathfrak{s}=\tilde{\mathfrak{s}}$ into Eq.~(\ref{eq:chgvar_rtozandsigma}), gives\footnote{We note that the condition $\cos (4b_1) = 1-2\mathfrak{z}_\pm$ is equivalent to $\sin^2 (2b_1) = \mathfrak{z}_\pm$ via basic trigonometry. }
\begin{equation}
\cos (4b_1) = 1-2\mathfrak{z}_-\quad\mbox{and}\quad\sin (4 b_1)=\tilde{\mathfrak{s}}\sqrt{1-(1-2\mathfrak{z}_-)^2}\,.
\end{equation}
Applying the trigonometric half-angle formulas 
 we can also write
\begin{equation}
\cos (2b_1) = \tilde{\mathfrak{s}}\sqrt{1-\mathfrak{z}_-}\quad\mbox{and}\quad\sin (2 b_1)=\sqrt{\mathfrak{z}_-}\,,
\end{equation}
and
\begin{equation}\label{eq:solfneq0}
\cos b_1 = \sqrt{\frac{1}{2}+\frac{\tilde{\mathfrak{s}}}{2}\sqrt{1-\mathfrak{z}_-}} \quad\mbox{and}\quad\sin b_1=\sqrt{\frac{1}{2}-\frac{\tilde{\mathfrak{s}}}{2}\sqrt{1-\mathfrak{z}_-}}\,.
\end{equation}
Since we generally have $\mathfrak{z}_- \in\;]0,1[$ and $\tilde{\mathfrak{s}}\in\{-1,+1\}$, the equations above well-define the value of the unitary parameter $b_1\in[0,\pi/2]$ for an arbitrary entangled X-state with $\mathcal{H}>\mathcal{G}$ and $h\neq 0$, which concludes the proof.

\end{document}